\documentclass[journal]{IEEEtran}
\usepackage[numbers,sort&compress]{natbib}
\usepackage{amsfonts}
\usepackage[cmex10]{amsmath}
\usepackage{array}
\usepackage{mdwmath}
\usepackage{mdwtab}
\usepackage{eqparbox}
\usepackage[caption=1]{caption}
\usepackage{stfloats}
\usepackage{url}
\usepackage{amssymb}
\usepackage{float}
\usepackage{indentfirst}
\usepackage{makeidx}
\usepackage{tabularx}
\usepackage{cases}
\usepackage{multirow}
\usepackage{diagbox}
\usepackage{mathtools}
\usepackage{color}
\usepackage{stfloats}
\usepackage{lipsum}
\usepackage{amsfonts}
\usepackage{bm}
\usepackage{dsfont}
\usepackage{lipsum}
\usepackage{amsthm}
\usepackage{graphicx, epsfig, subfigure}

\makeatletter

\newif\if@restonecol
\makeatother

\usepackage[linesnumbered,ruled,vlined]{algorithm2e}
\usepackage{algpseudocode}
\usepackage{amsmath, bm}

\begin{document}
\title{Distributed ADMM with Synergetic Communication and Computation}
\author{Zhuojun Tian,~
        Zhaoyang Zhang,~
        Jue Wang,~
        Xiaoming Chen,~
        Wei Wang,~
        and Huaiyu Dai
\thanks{Part of this work \cite{SCCD} was presented at ICNC 2020, Hawaii, USA.}
\thanks{This work was supported in part by National Key R\&D Program of China under Grant 2018YFB1801104, and National Natural Science Foundation of China under Grant 61725104 and 61631003, and Ningbo S\&T Major Project (No.~2019B10079).}
\thanks{Z. Tian (email: dankotian@zju.edu.cn), Z. Zhang (Corresponding Author, email: ning\_ming@zju.edu.cn), J. Wang, X. Chen, W. Wang are with the College of Information Science and Electronic Engineering, Zhejiang University, Hangzhou, China, and with Zhejiang Provincial Key Laboratory of Info. Proc., Commun. \& Netw. (IPCAN), Hangzhou 310027, China, and also with International Joint Innovation Center, Zhejiang University, Haining 314400, China. H. Dai (email: huaiyu\_dai@ncsu.edu) is with the Department of Electrical and Computer Engineering, NC State University, Raleigh, NC 27695, USA.}
}


\maketitle

\begin{abstract}
In this paper, we propose a novel distributed alternating direction method of multipliers (ADMM) algorithm with synergetic communication and computation, called SCCD-ADMM, to reduce the total communication and computation cost of the system. Explicitly, in the proposed algorithm, each node interacts with only part of its neighboring nodes, the number of which is progressively determined according to a heuristic searching procedure, which takes into account both the predicted convergence rate and the communication and computation costs at each iteration, resulting in a trade-off between communication and computation. Then the node chooses its neighboring nodes according to an importance sampling distribution derived theoretically to minimize the variance with the latest information it locally stores. Finally, the node updates its local information with a new update rule which adapts to the number of communication nodes.
We prove the convergence of the proposed algorithm and provide an upper bound of the convergence variance brought by randomness.
Extensive simulations validate the excellent performances of the proposed algorithm in terms of convergence rate and variance, the overall communication and computation cost, {the impact of network topology as well as the time for evaluation}, in comparison with the traditional counterparts.

%
\end{abstract}

\section{Introduction}
Nowadays, with the rapid development of Internet of Things (IoT), distributed information processing and decision making over networks have been highly demanded. In a typical distributed scenario, the original data or system parameters are often located in different agents which are supposed to collaboratively fulfill some global objective by communicating with others and computing over the information they have access. However, due to the usually limited communication and computation capacities of the agents in practice, the design of such distributed systems { is} extremely challenging.


One such well-known typical task is to solve the distributed optimization problem which widely exists in many areas such as machine learning and signal processing. In this problem, each node aims to optimize the global objective function through minimizing its local objective function and exchanging information with others, as represented in the following form:
\begin{equation}
	\min_{\bm{x}} \sum_{i=1}^{N} \phi_i(\bm{x}),
\end{equation}
where $\bm{x}\in\mathds{R}^{M}$ is the global variable to be optimized; $ \phi_i(\bm{x})=f_i(\bm{x})+\lambda r(\bm{x})$ is the local objective function of node $i$, which is composed of a smooth component $f_i(\bm{x})$ and a regularizer of $\ell_1$-norm or $\ell_2$-norm $r(\bm{x})$.


Basically, there are two types of algorithms to solve the above problem: gradient-based and dual decomposition-based.
Algorithms based on gradient or subgradient \cite{subg1, subg2, subg3} converge to the consensus optimal value by iteratively computing gradient or subgradient and then averaging among nodes. These algorithms play an important role in distributed optimization, but they have slow convergence rate in general. In contrast, dual decomposition based algorithms, like the  alternating direction method of multipliers (ADMM), can solve this problem with faster convergence by properly exploiting the problem structure \cite{admm}, and therefore have attracted lots of attention in recent years.


ADMM is initially implemented in a centralized network which has one central processor communicating with all agent nodes, aggregating their messages and processing the total information (see Fig. \ref{1a}). In such a scenario, the central node does not perform global average until receiving all the messages from all the nodes, which may lead to a long latency and lack of robustness to the processing errors of agent nodes. Compared with the centralized realization, the decentralized counterpart (see Fig. \ref{1b}) has no central processor and each node communicates with its neighboring nodes in a timely parallel manner, which fully exploits the network connectivity and makes the system potentially less sensitive to processing failures of agent nodes. Much research effort has been focused on the decentralized ADMM algorithm. Zhu et al. put forward a fully distributed decoding algorithm based on decentralized ADMM and verified its stability \cite{decode}. Mateos et al. applied the decentralized ADMM to deal with the linear regression problem \cite{regre}. In \cite{lrate}, the distributed ADMM is proven to converge linearly when the local objective function is strongly-convex,  while \cite{rate} shows that distributed converge with rate $\emph{O}(1/T)$ under a weaker assumption that the local function is convex. Zhang et al. extended the distributed ADMM algorithm to asynchrony scenario \cite{Asyn}, and Chang et al. provided the convergence analysis of asynchronous distributed ADMM in \cite{asyrate1, asyrate2}.

\begin{figure}[!ht]
\centering
\subfigure[centralized network]{
\includegraphics[width=0.23\textwidth]{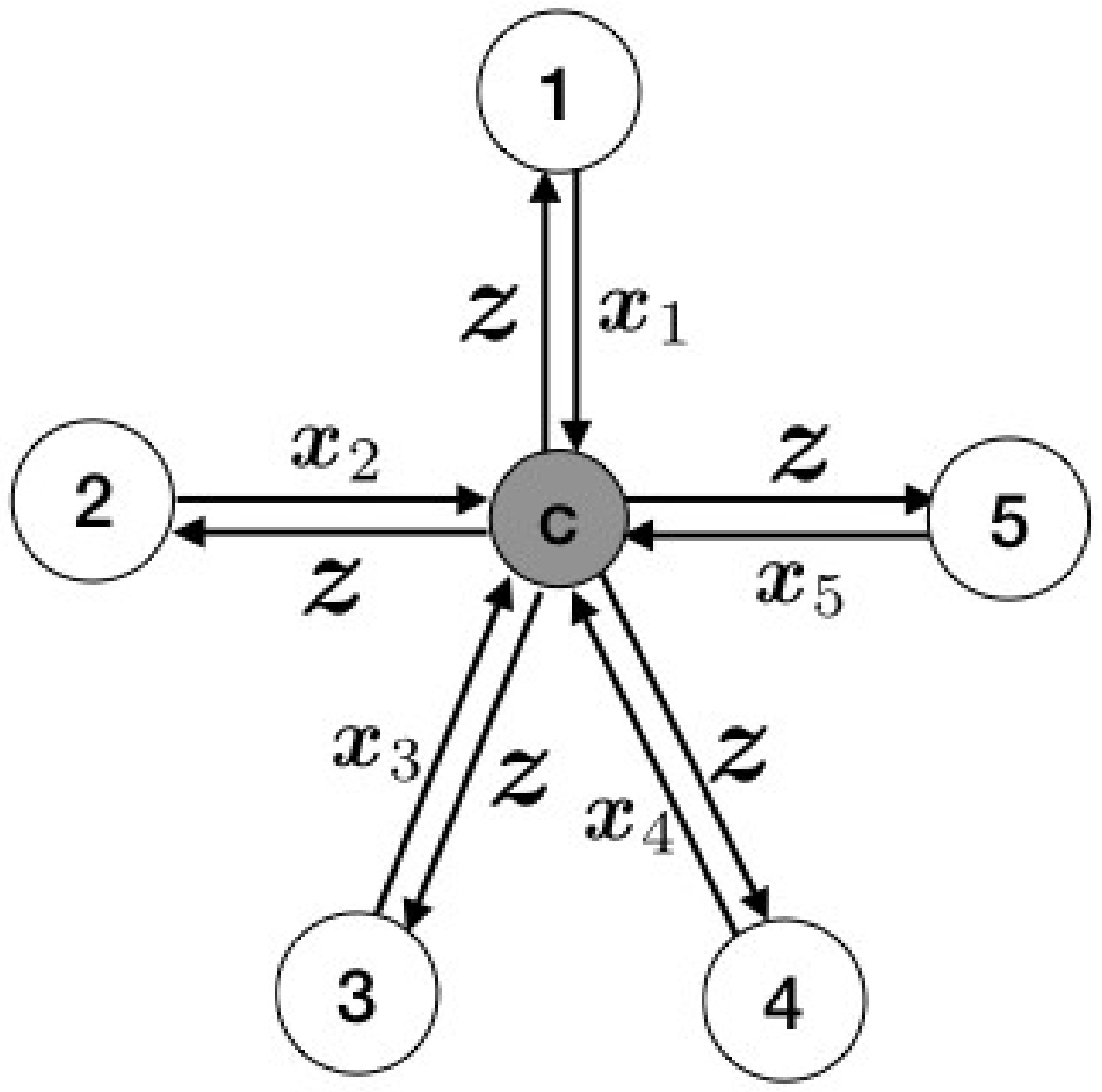}\label{1a}}
\subfigure[decentralized network]{
\includegraphics[width=0.23\textwidth]{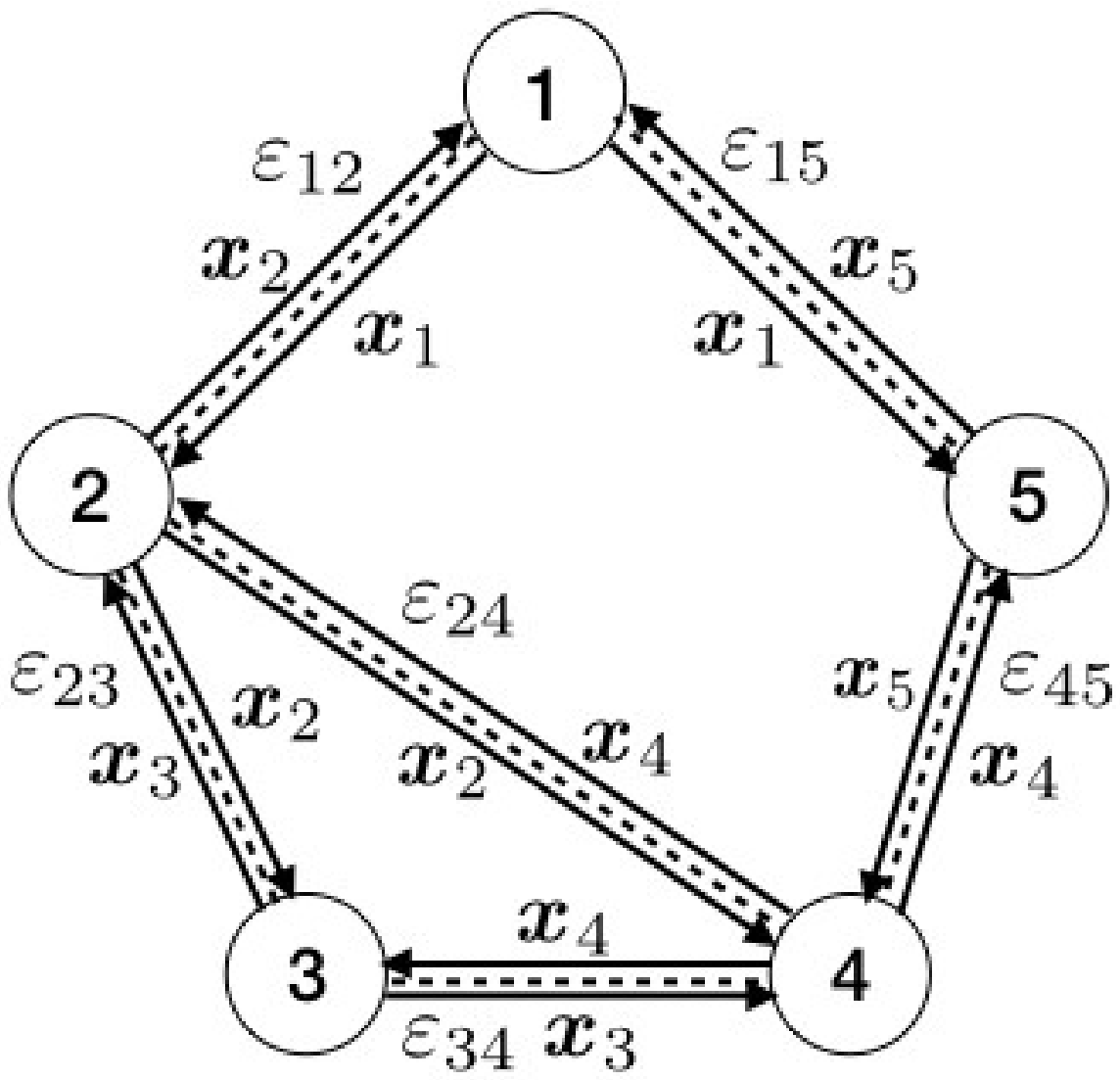}\label{1b}}
\caption{Illustration of the centralized and decentralized ADMM.}
\vspace{-0.3 cm}
\end{figure}

As both the network and data dimensions scale up, {reducing the transmission and processing cost has been the major concern in network protocol and algorithm design.
Much research effort has been put on saving the cost when implementing distributed ADMM under stringent resource constraints. This can be typically classified into two ways: reducing the number of iterations either theoretically or empirically and reducing the cost in each iteration with less transmission or less computation load.
Considering accelerating the algorithm and reducing the number of iterations, much work is based on tuning the penalty parameter, which can highly affect the efficiency of ADMM. One adaptive approach for penalty selection is consensus residual balancing (CRB) \cite{CRB1, CRB2}, which adjusts the penalty parameter so that the local derivatives of the Lagrangian w.r.t. primal and dual variables have similar magnitudes in each node.
On the other hand, some researchers focus on reducing the cost in each iteration. Two algorithms, IC-ADMM and IDC-ADMM, proposed by Chang et al. in \cite{Inexact} to alleviate the computation load, work in a simpler and more efficient way by using the proximal gradient. Zhu et al. gave a distributed ADMM with quantized communication to save the communication cost in \cite{Quantize_ADMM}, and the algorithms proposed in \cite{Wadmm, Communication-Censor} are based on communication censoring, reducing the communication links to save the cost.
The authors of \cite{Wadmm} proposed a weighted distributed ADMM to reduce communication usage through tuning the weight matrices of neighboring nodes. The algorithm maximizes the convergence speed under the constraint of communication arc number, which lessens the communication burden.
In \cite{Communication-Censor}, the communication-censored ADMM algorithm reduces the communication links by not allowing one node to transmit its local variable if the variable does not sufficiently differ from the last transmitted one.
}


In this paper, {we consider the problem of saving the total cost for distributed ADMM}.
In such a scenario of distributed ADMM, nodes need to exchange information frequently to achieve global consensus.
Intuitively, the communication cost can be saved by reducing the interaction between nodes. However, the resultant slower convergence rate and the larger number of iterations would cause the sharp increase in the computation cost, thus leading to a much higher total cost.
To save the total cost of the system, we propose a distributed ADMM with synergetic communication and computation (SCCD-ADMM), which makes the tradeoff between communication and computation while implementing distributed optimization.
Explicitly, each node only exchanges information with a fraction of its neighboring nodes. The number of communication nodes is determined by a heuristic searching procedure, which aims to reduce the communication cost as much as possible with acceptable computation cost.
Then the specific nodes to communicate with are chosen following the derived sampling distribution, and finally the node updates the local information with the newly-proposed update rule.

{Our work shares some similarities with \cite{Wadmm, Communication-Censor}, since we all focus on reducing the communication cost in each iteration.
However, the algorithm in \cite{Wadmm} assumes that the local function is convex and differentiable, meaning that it cannot deal with $\ell_1$-regularized objectives. Besides, the algorithm ignores the computation cost and the optimization of weight matrices is hard to implement in the practice. \cite{Communication-Censor} decides the transmission of one node by its variables' distance from the old version in the time dimension, while we consider selecting a subset of communication nodes to receive the messages, based on the difference of variables in the space dimension (among different nodes). Additionally, in SCCD-ADMM, the searching and sampling procedure makes the real communication network different in each iteration, which is a process of network topology learning to some extent. Similar idea appears in \cite{networklearning}, which focuses on providing faster model averaging for decentralized parallel stochastic gradient descent.}

The contributions of the paper can be summarized as follows:
\begin{itemize}
\item{We propose a distributed ADMM algorithm with new update rules, in which each node only interacts with part of its neighboring nodes. The amount of neighboring nodes is progressively determined according to a heuristic searching procedure, which takes into account both the predicted convergence rate and the communication and computation costs at each iteration, showing a trade-off between communication and computation.}
\item{We design a random distribution based on importance sampling and with the locally stored latest information for each node to choose the subset of neighboring nodes to communicate with. Based on this sampling criteria, the algorithm converges faster and promises a lower computation cost with the same communication cost.}
\item{We prove the convergence of the proposed algorithm and analyze the upper bound of the variance brought by randomness. We also provide extensive simulation which shows the resultant excellent performance in terms of convergence rate and variance, the overall communication and computation cost, the impact of network topology, the delay, etc.  }
\end{itemize}

This paper is organized as follows. In Section \ref{preliminary}, the optimization problem is formulated and the traditional distributed ADMM algorithm (D-ADMM) applied in the decentralized network is reviewed.
Section \ref{algorithm} elaborates on our proposed SCCD-ADMM algorithm, which is consist of three steps in each iteration. In Section \ref{analysis}, we give a theoretical analysis of the proposed algorithm in terms of convergence rate, the variance bound as well as the sampling distribution. Numerical experiments are given in Section \ref{experiment} to validate the convergence of the algorithm and the reduction of the total cost compared with the traditional way. {Moreover, the impact of the network topology and the delay is evaluated experimentally.}
Section \ref{conclusion} concludes the paper and discusses some possible future work.

{Note that this paper significantly extends our previous work \cite{SCCD} in several ways. Firstly, we give the theoretical analysis of the convergence property in both expectation and variance aspects, and derive the sampling distribution for each node, which verifies the intuitive distribution expression given in \cite{SCCD}. Secondly, a more comprehensive set of experiments is shown including the impact of searching stepsize, network topology as well as the delay comparison.}

\section{Preliminary}\label{preliminary}
\subsection{Network Model and Assumptions}
We consider a fully-distributed multi-agent network and represent it with an undirected network $\mathcal{G}=(\mathcal{V},\mathcal{E})$ as shown in Fig. 1{(b)}, where $\mathcal{V}=\{1,...,N\}$ denotes the set of $N$ nodes, and the edge set $\mathcal{E}=\{\varepsilon_{ij}\}_{i,j \in \mathcal{V}} $ indicates the communication links between nodes {as shown by dash lines in Fig. 1{(b)}}. We define the adjacency matrix of $\mathcal{G}$ as $\boldsymbol{W}$, where $[\boldsymbol{W}]_{i,j}=1$ if $\varepsilon_{ij}\in\mathcal{E}$ and $[\boldsymbol{W}]_{i,j}=0$ otherwise. $\mathcal{N}_i$ is the set of node $i$'s neighboring nodes and the diagonal degree matrix is defined as $\boldsymbol{D}=diag\{d_1,...,d_N\}$ and $d_i=|\mathcal{N}_i|$. The network is assumed to be connected, i.e., there exists a path between any pair of vertices.

{Each node aims to solve the problem (1) through communicating with its neighboring nodes. The unit communication cost, denoted by $C_\text{cmm}$, is defined as the energy consumption between two neighboring nodes to transfer a package of data. Likewise, the unit computation cost, denoted by $C_\text{cmp}$, is defined as the energy consumption of one node's updating for one package of data in one iteration. The total cost of the system is thus the cumulated communication and computation costs as the algorithm runs.
Three assumptions are made as follows:}

\newtheorem{assumption}{Assumption}
\begin{assumption}
The regularizer $r(\cdot)$ is convex.
\end{assumption}
\begin{assumption}
For all $i\in\mathcal{V}$, the function $f_i: \mathbb{R}^M\to\mathbb{R}$ in (1) is $\sigma$-strongly convex, i.e., there exists some $\sigma>0$ such that
\begin{equation}
f_i(\bm{u})\ge f_i(\bm{v})+\triangledown f_i(\bm{v})^T(\bm{u}-\bm{v})+\frac{\sigma}{2}\|\bm{u}-\bm{v}\|^2, \forall \bm{u},\bm{v}\in\mathbb{R}^M.
\end{equation}
In addition, $f_i$ is $(1/\gamma)$-smooth, which means that it has Lipschitz continuous gradients, i.e., there exists some $\gamma>0$ satisfying
\begin{equation}
f_i(\bm{u})\le f_i(\bm{v})+\triangledown f_i(\bm{v})^T(\bm{u}-\bm{v})+\frac{1}{2\gamma}\|\bm{u}-\bm{v}\|^2, \forall \bm{u},\bm{v}\in\mathbb{R}^M.
\end{equation}
\end{assumption}
\begin{assumption}
{The unit communication cost between each transmit-receive pair is the same and constant for all iterations. The unit computation cost is the same for each node and remains constant for all computation iterations.}
\end{assumption}
Assumption 1 can be satisfied for both $\ell_1$-norm and $\ell_2$-norm and Assumption 2 is necessary for the convergence of the algorithm. Assumption 3 guarantees that when choosing the neighboring nodes to communicate with, we can only consider the importance of the information of different neighboring nodes and neglect the discrepancy in the cost of communication links. The same computation cost in Assumption 3 is assumed for simplicity. {Note that the values of $C_\text{cmm}$ and $C_\text{cmp}$ are related to the number of bits transferred and the number of flops consumed. They are measured in advance in the practical implementation and thus are treated as known in this paper.}

\subsection{Traditional Distributed ADMM}
In the decentralized network, the structured formula (1) is not decomposable since the variable $\bm{x}$ is global. Thus the consensus variable is introduced into the optimization expression. Specifically, with $\bm{x}_i$ and $\bm{x}_j$ being the local variables in the node $i\in\mathcal{V}$ and its neighboring node $j\in{\mathcal{N}_i}$, the consensus variable $\bm{t}_{ij}$ is used to guarantee $\bm{x}_i=\bm{x}_j$, and then the local variables in different nodes through the network can converge to the equal values. On this basis, the optimization problem is reformulated as follows:
\begin{equation}
\begin{aligned}
	&\min_{\bm{x}_i}\quad \sum_{i=1}^{N} \Big [f_i(\bm{x}_i)+\lambda r(\bm{x}_i) \Big ], \\
	&s.t.{ \quad\bm{x}_i=\bm{t}_{ij},\bm{x}_j=\bm{t}_{ij} \quad\forall j\in {\mathcal{N}_i}, i\in\mathcal{V}.}
\end{aligned}
\end{equation}
Define $\phi_i(\bm{x}_i)\triangleq f_i(\bm{x}_i)+\lambda r(\bm{x}_i)$. Based on (2), we have the augmented Lagrange function as:
\begin{equation}
\begin{aligned}
\mathcal{L}&=\sum_{i=1}^N \phi_i(\bm{x}_i) +\sum_{i=1}^{N}\sum_{j\in{\mathcal{N}_i}}\Big [\bm{u}_{ij}^T(\bm{x}_i-\bm{t}_{ij})+\bm{v}_{ij}^T(\bm{x}_j-\bm{t}_{ij})\Big ]\\
&+\frac{c}{2}\sum_{i=1}^{N}\sum_{j\in{\mathcal{N}_i}}\Big [\|\bm{x}_i-\bm{t}_{ij}\|_2^2+\|\bm{x}_j-\bm{t}_{ij}\|_2^2\Big ],
\end{aligned}
\end{equation}
where $\bm{u}_{ij}$ and $\bm{v}_{ij}$ are Lagrange dual variables, $c$ is the penalty parameter and the last term is used to promote robustness.
At the $k$-th iteration and for $i\in\mathcal{V}$, define $\bm{p}_{i}^{(k)}\triangleq \sum_{j\in{\mathcal{N}_i}}(\bm{u}_{ij}^{(k)}+\bm{v}_{ji}^{(k)})$ with the initial condition $\bm{p}_{i}^{(0)}=\bm{0}$, then the update rule of the traditional distributed ADMM (i.e., D-ADMM) is \cite{regre}:
\begin{subequations}
\begin{flalign}\label{pupdateold}
\quad \bm{p}_i^{(k)}&=\bm{p}_{i}^{(k-1)}+c\sum_{j\in{\mathcal{N}_i}}(\bm{x}_i^{(k-1)}-\bm{x}_j^{(k-1)}),&
\end{flalign}
\begin{flalign}\label{xupdateold}
\begin{split}
\quad \bm{x}_{i}^{(k)}=&\mathop{\arg\min}_{\bm{x}_i} \ \ \{f_i(\bm{x}_i)+\lambda r(\bm{x}_i)+\bm{x}_{i}^T\bm{p}_{i}^{(k)}+\\
&c\sum_{j\in{\mathcal{N}_i}}\|{\bm{x}_i-\frac{\bm{x}_i^{(k-1)}+\bm{x}_j^{(k-1)}}{2}}\|_2^2\}, \forall i\in\mathcal{V},
\end{split}&
\end{flalign}
\end{subequations}
where $c$ is the penalty parameter.

In each iteration, each node updates its variables according to (6) in a fully distributed way, after which it communicates with all its neighboring nodes to exchange the update results.
The algorithm iterates until it converges and the problem is globally optimized.
However, this traditional distributed ADMM algorithm requires each node to communicate with all of its neighbors, which is inefficient in the large scale scenario {in terms of energy consumption.}

{
Our proposed algorithm aims to improve the energy efficiency of the distributed ADMM algorithm and save the total cost of the system by reducing communication nodes at each iteration. There are several challenges to overcome. Firstly, since fewer communication nodes may lead to more iterations and increased overall computation cost, how to conduct the tradeoff between communication and computation, or namely, how to decide the number of communication nodes on the premise of an acceptable convergence rate, should be considered. Secondly, how to select the nodes to communicate with given the number of nodes and how to re-design the update rule suitable for them, should also be considered.
In the next section, we will develop an algorithm to solve the problems.}

\section{Algorithm Derivation} \label{algorithm}
In our proposed algorithm, each node follows three steps in each iteration: searching the number of communication nodes, selecting the nodes to communicate with and updating the local variables. Explicitly, each node first decides an appropriate node number by implementing a local searching procedure which attempts to minimize the total cost based on the latest updates received from this number of neighboring nodes. Then the node selects this number of neighboring nodes according to a certain importance sampling distribution which is also calculated using the latest updates received. Finally, after communicating with these nodes and getting their updates, the node updates its own variable estimates based on new update rules.
In the following, for better presentation, we first give the update rules of the algorithm, which is the basis of the work. After that, the searching process and the selecting criterion are described in detail.

\subsection{Update Rules}
To lessen the communication load in the proposed SCCD-ADMM, the node $i$ receives messages from a subset of its neighboring nodes, instead of all of them as in D-ADMM, to update its estimation of the variables. However, the update rules have to be changed accordingly to ensure convergence.

Given the number of the communication nodes $Num_i$ and the set of the chosen nodes $\mathcal{N}_{\text{c},i} \subseteq {\mathcal{N}_i}$, $|\mathcal{N}_{\text{c},i}|= Num_i$, we first consider the update rule for $\bm{p}_i$.
When node $j\in {\mathcal{N}_{\text{c},i}}$ is selected to transfer message to node $i$ in the $k$-th iteration, with this message the gradient term in (\ref{pupdateold}) can be estimated by
$$\frac{(\bm{x}_i^{(k-1)}-\bm{x}_j^{(k-1)})}{w_{ij}^{(k)}},$$
where $w_{ij}^{(k)}$ is the probability for node $i$ to select node $j$ to communicate with at iteration $k$ as determined in Section \ref{ChooseNode}. Note that this estimation is unbiased since
\begin{equation}\nonumber
\begin{aligned}
\mathbb{E}\Big[\frac{\bm{x}_i^{(k-1)}-\bm{x}_j^{(k-1)}}{w_{ij}^{(k)}}\Big]
&=\sum_{j\in\mathcal{N}_{i}}w_{ij}^{(k)}\Big[\frac{\bm{x}_i^{(k-1)}-\bm{x}_j^{(k-1)}}{w_{ij}^{(k)}}\Big]\\
&=\sum_{j\in\mathcal{N}_{i}}(\bm{x}_i^{(k-1)}-\bm{x}_j^{(k-1)}).
\end{aligned}
\end{equation}
Then averaging over all $Num_i$ messages received from the selected set of neighboring nodes $\mathcal{N}_{\text{c},i}$, $\bm{p}_i$ can be updated in a gradient descent way as follows:
\begin{subequations}
\begin{equation}\label{pupdate}
	\bm{p}_{i}^{(k)}=\bm{p}_{i}^{(k-1)}+\frac{c}{Num_i}\sum_{j\in{\mathcal{N}_{\text{c},i}}}\frac{(\bm{x}_i^{(k-1)}-\bm{x}_j^{(k-1)})}{w_{ij}^{(k)}}.
\end{equation}

Similarly, when it comes to the $\bm{x}$-update, as in the $\bm{p}$-update, each node only receives the messages from the selected subset of neighboring nodes and thus the last term in (\ref{xupdateold}) turns into
\begin{equation}\label{hupdate}
h_{i,j}^{(k)}(\bm{x}_i)=\frac{c}{Num_i}\sum_{j\in{\mathcal{N}_{\text{c},i}}}\frac{\|\bm{x}_i-\frac{1}{2}(\bm{x}_i^{(k-1)}+\bm{x}_j^{(k-1)})\|_2^2}{w_{ij}^{(k)}}.
\end{equation}
To reduce the deviation caused by the randomness in node selection and increase the convergence performance of the algorithm, inspired by the idea in \cite{Imsam}, we replace $h_{i,j}^{(k)}(\bm{x}_i)$ with its first-order prediction at the $(k-1)$-th iteration, i.e., $\nabla h_{i,j}^{(k)}(\bm{x}_i^{(k-1)})\times\bm{x}_i$, and add a quadratic term $\frac{1}{\eta_k}\|\bm{x}_i-\bm{x}_{i}^{(k-1)}\|_2^2$ to the $\bm{x}$-update. Note that as the algorithm runs, $h_{i,j}^{(k)}(\bm{x}_i)$ approaches $0$, so a first-order approximation is sufficient. The quadratic term is a { proximal term}, which restricts the update result $\bm{x}_i^{(k)}$ to be close to $\bm{x}_i^{(k-1)}$ and thus helps to lower the variance and promise the convergence of the algorithm. The parameter ${\eta_k}$ in the term can affect the convergence of the algorithm and we set it as \cite{Stoadmm}: ${\eta_k}=\frac{D}{\sqrt{2k}}$, where $D$ is a constant and its value is given experimentally and $k$ is the iteration number.

On this basis, the update rule of $\bm{x}$ is as follows:
\begin{equation}
\begin{aligned}\label{xupdate}
	\bm{x}_{i}^{(k)}&=\mathop{\arg\min}_{\bm{x}_i} \Big \{f_i(\bm{x}_i)+\lambda r(\bm{x}_i)+\bm{x}_i^{T}\bm{p}_i^{(k)}+\\
	&\nabla h_{i,j}^{(k)}(\bm{x}_i^{(k-1)})\times\bm{x}_i + \frac{\|\bm{x}_i-\bm{x}_{i}^{(k-1)}\|_2^2}{\eta_k}\Big \}.
\end{aligned}
\end{equation}
\end{subequations}

The update of $\bm{x}_{i}^{(k)}$ in (\ref{xupdate}) can be easily accomplished. Specifically, when the objective is regularized by $\ell_2$-norm, i.e., $r(\bm{x}_i)=\|\bm{x}_i\|_2$, $\bm{x}_{i}^{(k)}$ can be obtained using gradient descent way since the optimization {objective} of (\ref{xupdate}) is convex and differential. When it comes to $\ell_1$-norm, i.e., $r(\bm{x}_i)=\|\bm{x}_i\|_1$, problem (\ref{xupdate}) can also be solved efficiently with fast iterative shrinkage thresholding algorithm (FISTA) \cite{FISTA}.


\subsection{Decide the Number of Communication Nodes}\label{DecideNumber}
Although a smaller number of communication nodes can save the {iteration-wise} communication cost, {it cannot necessarily save the overall communication cost. Instead, the resultant computation cost may be much higher due to larger iteration number, thus causing an unbearable total cost.} Hence, it is essential for each node to make a dedicate decision $Num_i$, the amount of communication nodes, with the aim to minimize the total communication and computation cost or trade off between them.
{A direct method maybe to optimize $Num_i$ and fix it during the iterations. However, it is hard to choose the well-performed number due to the difficulty of deriving the convergence rate w.r.t $Num_i$.
Moreover, the properties and the distributions of the variables change continuously as the iteration goes on, and thus it is essential to design an adaptive searching method for capturing the dynamic behavior and achieving a better trade-off between the  communication and computation cost.}
To this end, a searching procedure is designed, in which each node makes progressive search attempts for different numbers of neighboring nodes by evaluating a well-designed measure function which reflects the overall computation and communication cost.



Note that the overall computation and communication cost in fact involves two aspects: 1) the number of iterations needed to achieve the convergence, 2) the computation and communication cost paid for each iteration. Although the convergence rate of the algorithm is hard to be precisely calculated, it can be roughly predicted by the decrease of the objective value or the consensus error in each iteration, i.e., a larger decrease of the objective value or consensus error in one iteration means faster convergence or less iterations needed to converge.

Keeping this in mind, firstly, we measure the convergence tendency of the algorithm for node $i$ in the $k$-th iteration by $\Gamma_i^{(k)}(Num^{(s)}_i)$ in which $Num^{(s)}_i$ is the attempted communication node number at its $s$-th search attempt. {Note that the convergence tendency depends on the one changing more slowly between the accuracy and the consensus error.} $\Gamma_i^{(k)}(Num^{(s)}_i)$ is determined as follows:


\begin{itemize}
\item{When the objective is regularized by $\ell_2$ norm, {the main constraint for convergence is the accuracy decided by the objective value. So}
$\Gamma_i^{(k)}(Num^{(s)}_i)$ is defined as the predicted decrease of the objective value w.r.t.  the estimate at the previous iteration, i.e., $\Gamma_i^{(k)}(Num^{(s)}_i)=\vert\phi_i(\bm{x}_i^{(k-1)})-\phi_i(\hat{\bm{x}}_i^{(k)}(Num^{(s)}_i))\vert$, where $\vert\cdot\vert$ denotes taking the absolute value to cope with the seldom fluctuation, and $\hat{\bm{x}}_i^{(k)}(Num^{(s)}_i)$ is the estimated variables at the $k$-th iteration using the attempted node number $Num^{(s)}_i$.}
\item{When the regularizer is $\ell_1$ norm, {the consensus error becomes the criterion changing more slowly due to the sparsity introduced by $\ell_1$ norm. So}
$\Gamma_i^{(k)}(Num^{(s)}_i)$ is the predicted decrease of the consensus error, i.e., $\Gamma_i^{(k)}(Num^{(s)}_i)=\vert \|\bar{\bm{x}}_{\mathcal{N}_i}^{(k-1)}-{\bm{x}}_i^{(k-1)}\|_2 - \|\bar{\bm{x}}_{\mathcal{N}_i}^{(k)}-\hat{\bm{x}}_i^{(k)}(Num^{(s)}_i)\|_2\vert$, where $\bar{\bm{x}}_{\mathcal{N}_i}^{(k)}$ denotes the average of all neighboring nodes' values.}
\end{itemize}
To sum up, the convergence measure is expressed as follows:
\begin{equation}\label{Gamma}
\begin{aligned}
&\Gamma_i^{(k)}(Num^{(s)}_i)\\
&=\left\{
\begin{array}{rl}
\vert\phi_i(\bm{x}_i^{(k-1)})-\phi_i(\hat{\bm{x}}_i^{(k)}(Num^{(s)}_i))\vert, \\\text{for}~{r(\cdot)=\|\cdot\|_2};\\
\vert \|\bar{\bm{x}}_{\mathcal{N}_i}^{(k-1)}-{\bm{x}}_i^{(k-1)}\|_2 - \|\bar{\bm{x}}_{\mathcal{N}_i}^{(k)}-\hat{\bm{x}}_i^{(k)}(Num^{(s)}_i)\|_2\vert,\\ \text{for}~{r(\cdot)=\|\cdot\|_1}.
\end{array}\right.
\end{aligned}
\end{equation}

Note that (\ref{Gamma}) depends on the updated result $\hat{\bm{x}}_i^{(k)}(Num^{(s)}_i)$ which relies on updates received from $Num_i^{(s)}$ neighboring nodes.
However, if node $i$ actually receives messages from its neighboring nodes in every search attempt, it will lead to high communication cost. In order to avoid extra communication cost in the searching process, we let node $i$ choose $Num_i^{(s)}$ neighboring nodes with the sampling distribution provided in Section \ref{ChooseNode} and update $\hat{\bm{x}}_i^{(k)}(Num^{(s)}_i)$ with the latest locally-stored information of these nodes.
Likewise, $\bar{\bm{x}}_{\mathcal{N}_i}^{(k)}$ in (\ref{Gamma}) is also estimated using the latest information stored in node $i$. {In this way, the searching procedure has no communication complexity.}

Secondly, we measure the computation and communication costs in each iteration. The computation cost mainly results from the updating processes in the searching procedure, which include both the $\bm{x}$-update and the $\bm{p}$-update and are supposed to be of the same computation cost for all nodes according to Assumption 3.
Since in each search attempt the computation cost mainly comes from the calculation of $\hat{\bm{x}}_i^{(k)}(Num^{(s)}_i)$, if node $i$ stops its searching at the $s$-th attempt in the $k$-th iteration, it results in $(s+1)\times C_\text{cmp}$ units of computation cost, taking into account the final $\bm{x}$-update in that iteration.
It is easy to estimate the communication cost. Using $C_{\text{cmm},i}$ to denote the total communication cost for node $i$ in one iteration, again, with Assumption 3, we have $C_{\text{cmm},i}= Num_i^{(s)} \times C_\text{cmm}$.
{Note that in the practical implementation, the value of $C_\text{cmp}$ and $C_\text{cmm}$ are measured in advance to assist the searching process.}

Now we discuss the overall evaluation function. As mentioned above, $\Gamma_i^{(k)}(Num^{(s)}_i)$ reflects to some extent the convergence rate at each search attempt. More concretely, the predicted number of iterations needed under the current setting of communication node number is roughly proportional to the inverse of $\Gamma_i^{(k)}(Num^{(s)}_i)$. Therefore, the overall evaluation function used to determine the best number of communication nodes can be defined as:
\begin{equation}\label{evalfunction}
	E_i^{(k)}(Num_i^{(s)})\triangleq \frac{(s+1)\times C_\text{cmp} + Num_i^{(s)}\times C_\text{cmm}}{\Gamma_i^{(k)}(Num^{(s)}_i)}.
\end{equation}

Finally, the searching process is described as follows: in each iteration, node $i$ makes a series of attempts to find an appropriate communication node number in a descending manner. In the $s$-th searching attempt, node $i$ chooses $Num_i^{(s)}$ neighboring nodes according to the process described in Section \ref{ChooseNode}, and then calculates the corresponding evaluation function (\ref{evalfunction}) based on the latest locally stored messages from these nodes. The attempts continue with reduced $Num_i^{(s)}$ until the evaluation function starts to increase, i.e., $E_i(Num_i^{(s)}) > E_i(Num_i^{(s-1)})$, or the node number could not decrease any more. Note that to ensure the independency of information transferred among nodes as in message-passing based algorithms, for each iteration and in average, each edge is chosen by one of its two nodes, which means that the initial searching number of communication nodes for node $i$ can be set as $|\mathcal{N}_i|/2$. Also note that this stopping criterion of the searching may not lead to global optimum because the function $E_i(Num_i^{(s)})$ could be nonconvex. However, since both the numerator and the denominator in (\ref{evalfunction}) generally increase with $Num_i^{(s)}$, such a search at least has a good chance to find a local optimum even which can exhibit a good performance as shown in our experiments (see Section \ref{experiment}).

We summarize the searching procedure in \textbf{Algorithm \ref{searchnum}}.
\begin{algorithm}
\caption{\textbf{Searching Procedure}}\label{searchnum}
\textbf{Input:}  the iteration number $k$; $Num_i$ in last iteration.\\
\textbf{Output:} the number of search attempts $s$; $Num_i$ in current iteration.\\
$s \leftarrow 0$.\\
\If{$k=1$}{
	$Num_i^{(s)}\leftarrow |\mathcal{N}_i|/2$.\\
	}
\Else{
    $Num_i^{(s)}\leftarrow Num_i$.\\
    }
$s\leftarrow s+1$,\\
$Num_i^{(s)}\leftarrow {\text{max}(Num_i^{(s-1)}- stepsize, 1)}$.\\
\While {$E^{(k)}_i(Num_i^{(s)})\le E^{(k)}_i(Num_i^{(s-1)})$}{
			$s\leftarrow s+1$,\\
			$Num_i^{(s)}\leftarrow  {\text{max}(Num_i^{(s-1)}- stepsize, 1)}$.\\
		}
$Num_i \leftarrow Num_i^{(s)}$.\\
\end{algorithm}

\subsection{Choose the Communication Nodes} \label{ChooseNode}
Given the number of communication nodes, node $i$ needs to figure out which nodes to receive information from.
The selection of the communication nodes is treated as a sampling process, i.e., node $i$ samples from the set of its neighboring nodes.
In the sequel, we first consider how to choose one node to communicate with and then generalize it to the case of more than one nodes.
A widely-used sampling method is uniform sampling, promising an unbiased estimation. However, it may lead to high variance and negatively influence the convergence. So instead of uniform sampling, we consider using importance sampling technique. Importance sampling is an active sampling method applied extensively in stochastic optimization. Zhao et al. \cite{Imsam} proved that when the sampling distribution is proportional to the norm of stochastic gradient, the variance can be minimized. However, in our studied scenario, this cannot be directly applied because the update of $\bm{x}_i$ is not a simple realization of gradient descent.

Here, by minimizing the variance of the updating result, we get the sampling distribution as
\begin{equation}
\label{distribution1}
w_{ij}^{(k)}=\frac{\|\bm{x}_i^{(k-1)}-\bm{x}_j^{(k-1)}\|_2}{\sum_{l\in\mathcal{N}_{i}}\|\bm{x}_i^{(k-1)}-\bm{x}_j^{(k-1)}\|_2},
\end{equation}
and the detailed derivation can be found in Theorem 3 in Section \ref{analysis}.

In the expression of $w_{ij}^{(k)}$, the $\bm{x}$-update result of each neighboring node $j\in \mathcal{N}_{i}$ in the $(k-1)$-th iteration , i.e., $\bm{x}_j^{(k-1)}$, is required. However, node $i$ only communicates with part of the neighboring nodes in each iteration and it does not have all other neighboring nodes' information. To solve this problem, node $i$ uses the locally-stored information received from this part of neighboring nodes, in a similar way as depicted in Section \ref{DecideNumber}.
It means that when estimating the distribution, node $i$ uses the latest information received from all of its neighboring nodes.
On this basis, we can modify the expression of $w_{ij}^{(k)}$ as:
\begin{equation}\label{distribution}
w_{ij}^{(k)}=\frac{\|\bm{x}_i^{(k-1)}-\bm{x}_j^{(l_j)}\|_2}{\sum_{l\in\mathcal{N}_{i}}\|\bm{x}_i^{(k-1)}-\bm{x}_j^{(l_j)}\|_2},
\end{equation}
where $\bm{x}_j^{(l_j)}$ represents the latest stored messages received from node $j$.
{Given the sampling distribution, each node needs to send the on-off triggers to its communication nodes to activate the transmission.}


The above procedures compose the proposed SCCD-ADMM algorithm which is sketched in \textbf{Algorithm \ref{sccdadmm}}. Note that it stops when the accuracy and consensus error \cite{Inexact} are lower than some thresholds. The accuracy of the algorithm is measured by $acc=(obj(\bar{\bm{x}}^{(k)})-obj^{*})/obj^{*}$ and the consensus error of all nodes' optimization results is defined as $cserr=\sum_{i=1}^N\|\bar{\bm{x}}^{(k)}-\bm{x}_i^{(k)}\|_2^2/N$, where $obj^{*}$ is the optimal value of the objective and $\bar{\bm{x}}^{(k)}=(\sum_{i=1}^{N}\bm{x}_i^{(k)})/N$.

\begin{algorithm}
\caption{\textbf{SCCD-ADMM}}\label{sccdadmm}
\For{all node $i\in\mathcal{V}$ [in parallel]}
{
\textbf{Initialize:} $\bm{p}_i^{(0)}=\bm{0}$, $\bm{x}_i^{(0)}=\bm{0}$, $k=0$.
}
\While{not stopping}
{
$k=k+1$.\\
	\For{all node $i\in\mathcal{V}$ [in parallel]}
	{
	Decide the number of communication nodes $Num_i$ (Implement \textbf{Algorithm \ref{searchnum}}).\\
	Choose communication nodes $\mathcal{N}_{\text{c},i}$ according to (\ref{distribution}).\\
	Update $\bm{p}_i^{(k)}$ according to (\ref{pupdate}).\\
	Update $\bm{x}_i^{(k)}$ according to (\ref{xupdate}).
	}
}
\end{algorithm}

\section{Performance Analysis}\label{analysis}
In this section, we analyze the convergence of the proposed SCCD-ADMM algorithm. In addition, we establish the variance of the stochastic updating result and provide the upper bound of the convergence variance, on the basis of which we derive the sampling distribution for each node to choose its communication nodes and then give the corresponding error bound.
%

{Before illustrating Theorem 1, {we define the feature matrix $\boldsymbol{A}_i\in \mathbb{R}^{M\times m}$ of node $i$ as a data matrix composed of the coefficients of the variables $\bm{x}\in\mathbb{R}^{M}$. If there is only one set of coefficients of $\bm{x}$, then $m=1$. Otherwise, like in a classical regression problem, $m$ equals to the number of training samples. Here each column of $\boldsymbol{A}_i$ is one feature vector.}
{Under this definition, we firstly consider the convergence of the proposed SCCD-ADMM update rule in the sense of expectation.}

\newtheorem{theorem}{Theorem}
\begin{theorem}\label{theorem_convergence}
Let $\frac{2}{\eta_k}>c\lambda_{max}$, where $\lambda_{max}$ denotes the largest eigenvalue of $(\boldsymbol{D}-\boldsymbol{W})$.
$\boldsymbol{A}_i$ is the feature matrix of the training data in node $i$.
Let $\bm{x}^*\triangleq\bm{x}_1^*=...=\bm{x}_N^*$ and $\{\bm{u}_{ij}^*,\bm{v}_{ij}^*,j\in\mathcal{N}_i\}_{i=1}^N$ be a pair of optimal and dual solutions to the optimization problem. Under the assumptions listed before, we have

\emph{a) } $\bm{x}_1^{(k)},...,\bm{x}_N^{(k)}$ converge to the optimal point $\bm{x}^*$ in expectation.

\emph{b) } If $\boldsymbol{A}_i$ has full column rank, for all node in the network, the variables converge to the optimal value linearly in expectation.
\end{theorem}

\begin{proof}
According to the KKT conditions, for $\forall i\in\mathcal{V}$, we have:
\begin{subequations}
\begin{equation}\label{13a}
\nabla f_i(\bm{x}_i^*)+\lambda\partial r(\bm{x}_i^*)+\sum_{j\in\mathcal{N}_i}(\bm{u}_{ij}^*+\bm{v}_{ji}^*)=0,
\end{equation}
\begin{equation}\label{13b}
\bm{x}_i^*=\bm{x}_j^*, \forall j\in\mathcal{N}_i,
\end{equation}
\begin{equation}\label{13c}
\bm{u}_{ij}^*+\bm{v}_{ij}^*=0, \forall j\in\mathcal{N}_i,
\end{equation}
\end{subequations}
where $\partial r(\bm{x}_i^*)$ denotes the subgradient of $r$ at $\bm{x}_i^*$. When the algorithm converges and consensus among agents is achieved, we have $\bm{x}^*\triangleq\bm{x}_i^*,\forall i\in\mathcal{V}$ and $\tilde{\bm{x}}^*=\bm{1}_N\otimes\bm{x}^*$.

We first consider the condition of $Num_i=1$ and suppose we choose node $l$ as the communication node. According to our definition that $\bm{p}_{i}^k\triangleq \sum_{j\in{\mathcal{N}_i}}(\bm{u}_{ij}^{(k)}+\bm{v}_{ji}^{(k)}),  i\in\mathcal{V}$ and considering the optimal condition of (\ref{xupdate}), we have that
\begin{equation}\label{14}
\begin{aligned}
&\nabla f_i(\bm{x}_i^{(k)})+\lambda\partial r(\bm{x}_i^{(k)})+\frac{2}{\eta_k}(\bm{x}_i^{(k)}-\bm{x}_i^{(k-1)})\\
&+\sum_{j\in\mathcal{N}_i}(\bm{u}_{ij}^{(k)}+\bm{v}_{ji}^{(k)})+\frac{c}{w_{il}^{(k)}}(\bm{x}_i^{(k-1)}-\bm{x}_l^{(k-1)})=0.
\end{aligned}
\end{equation}
(\ref{14}) minus (\ref{13a}) and we get
\begin{equation}\label{15}
\begin{aligned}
&\nabla f_i(\bm{x}_i^{(k)})-\nabla f_i(\bm{x}_i^*)+\frac{2}{\eta_k}(\bm{x}_i^{(k)}-\bm{x}_i^{(k-1)})\\
&+\lambda\partial r(\bm{x}_i^{(k)})-\lambda\partial r(\bm{x}_i^*)+\sum_{j\in\mathcal{N}_i}(\bm{u}_{ij}^{(k)}+\bm{v}_{ji}^{(k)}-\bm{u}_{ij}^*-\bm{v}_{ji}^*)\\
&+\frac{c}{w_{il}^{(k)}}(\bm{x}_i^{(k-1)}-\bm{x}_l^{(k-1)})=0.
\end{aligned}
\end{equation}
We multiply $(\bm{x}_i^{(k)}-\bm{x}_i^{*})$ on both sides and reformulate (\ref{15}) as
\begin{equation}\label{16}
\begin{aligned}
&(\nabla f_i(\bm{x}_i^{(k)})-\nabla f_i(\bm{x}_i^{*}))^T(\bm{x}_i^{(k)}-\bm{x}_i^{*})\\
&+\frac{2}{\eta_k}(\bm{x}_i^{(k)}-\bm{x}_i^{(k-1)})^T(\bm{x}_i^{(k)}-\bm{x}_i^{*})\\
&+\lambda(\partial r(\bm{x}_i^{(k)})-\partial r(\bm{x}_i^*))^T(\bm{x}_i^{(k)}-\bm{x}_i^{*})\\
&+\sum_{j\in\mathcal{N}_i}(\bm{u}_{ij}^{(k)}+\bm{v}_{ji}^{(k)}-\bm{u}_{ij}^*-\bm{v}_{ji}^*)^T(\bm{x}_i^{(k)}-\bm{x}_i^{*})\\
&+\frac{2c}{w_{il}^{(k)}}(\bm{x}_i^{(k)}-\frac{\bm{x}_i^{(k-1)}+\bm{x}_l^{(k-1)}}{2})^T(\bm{x}_i^{(k)}-\bm{x}_i^{*})\\
&+\frac{2c}{w_{il}^{(k)}}(\bm{x}_i^{(k-1)}-\bm{x}_i^{(k)})^T(\bm{x}_i^{(k)}-\bm{x}_i^*)=0.
\end{aligned}
\end{equation}
We take expectation on both sides. Considering the update rule of $\bm{p}_i$, the forth term is equal to
\begin{equation}\label{17}
\begin{aligned}
&\mathbb{E}[\sum_{j\in\mathcal{N}_i}(\bm{u}_{ij}^{(k)}+\bm{v}_{ji}^{(k)}-\bm{u}_{ij}^*-\bm{v}_{ji}^*)^T(\bm{x}_i^{(k)}-\bm{x}_i^{*})]\\
&=\mathbb{E}[\sum_{j\in\mathcal{N}_i}(\bm{u}_{ij}^{(k+1)}+\bm{v}_{ji}^{(k+1)}-\bm{u}_{ij}^*-\bm{v}_{ji}^*)^T(\bm{x}_i^{(k)}-\bm{x}_i^{*})]\\
&-\mathbb{E}[\frac{2c}{w_{il}^{(k+1)}}(\bm{x}_i^{(k)}-\frac{\bm{x}_i^{(k)}+\bm{x}_l^{(k)}}{2})^T(\bm{x}_i^{(k)}-\bm{x}_i^{*})]\\
&=\mathbb{E}[\sum_{j\in\mathcal{N}_i}(\bm{u}_{ij}^{(k+1)}+\bm{v}_{ji}^{(k+1)}-\bm{u}_{ij}^*-\bm{v}_{ji}^*)^T(\bm{x}_i^{(k)}-\bm{x}_i^{*})]\\
&-2c\sum_{j\in\mathcal{N}_i}(\bm{x}_i^{(k)}-\frac{\bm{x}_i^{(k)}+\bm{x}_j^{(k)}}{2})^T(\bm{x}_i^{(k)}-\bm{x}_i^{*}).
\end{aligned}
\end{equation}
The expectation of the fifth term of (\ref{16}) is
\begin{equation}\label{18}
\begin{aligned}
&\mathbb{E}[\frac{2c}{w_{il}^{(k)}}(\bm{x}_i^{(k)}-\frac{\bm{x}_i^{(k-1)}+\bm{x}_l^{(k-1)}}{2})^T(\bm{x}_i^{(k)}-\bm{x}_i^{*})]\\
&=2c\sum_{j\in\mathcal{N}_i}(\bm{x}_i^{(k)}-\frac{\bm{x}_i^{(k-1)}+\bm{x}_j^{(k-1)}}{2})^T(\bm{x}_i^{(k)}-\bm{x}_i^{*}),
\end{aligned}
\end{equation}
and the expectation of the sixth term is
\begin{equation}\label{19}
\begin{aligned}
&\mathbb{E}[\frac{2c}{w_{il}^{(k)}}(\bm{x}_i^{(k-1)}-\bm{x}_i^{(k)})^T(\bm{x}_i^{(k)}-\bm{x}_i^*)]\\
&=2c\sum_{j\in\mathcal{N}_i}(\bm{x}_i^{(k-1)}-\bm{x}_i^{(k)})^T(\bm{x}_i^{(k)}-\bm{x}_i^*)\\
&=2c|\mathcal{N}_i|(\bm{x}_i^{(k-1)}-\bm{x}_i^{(k)})^T(\bm{x}_i^{(k)}-\bm{x}_i^*).
\end{aligned}
\end{equation}
The condition of $Num_i>1$ is similar because the expectation of the forth, fifth term and sixth term is equal to (\ref{17}-\ref{19}).

The other terms remain the same. So after taking expectations, we can get the following
\begin{equation}\label{20}
\begin{aligned}
\mathbb{E}[&(\nabla f_i(\bm{x}_i^{(k)})-\nabla f_i(\bm{x}_i^{*}))^T(\bm{x}_i^{(k)}-\bm{x}_i^{*})\\
&+\lambda(\partial r(\bm{x}_i^{(k)})-\partial r(\bm{x}_i^*))^T(\bm{x}_i^{(k)}-\bm{x}_i^{*})\\
&+(\frac{2}{\eta_k}-2c|\mathcal{N}_i|)(\bm{x}_i^{(k)}-\bm{x}_i^{(k-1)})^T(\bm{x}_i^{(k)}-\bm{x}_i^{*})\\
&+\sum_{j\in\mathcal{N}_i}(\bm{u}_{ij}^{(k+1)}+\bm{v}_{ji}^{(k+1)}-\bm{u}_{ij}^*-\bm{v}_{ji}^*)^T(\bm{x}_i^{(k)}-\bm{x}_i^{*})\\
&2c\sum_{j\in\mathcal{N}_i}(\frac{\bm{x}_i^{(k)}+\bm{x}_j^{(k)}}{2}-\frac{\bm{x}_i^{(k-1)}+\bm{x}_j^{(k-1)}}{2})]=0,
\end{aligned}
\end{equation}

By the strong convexity of $f_i$ and the convexity of $r$, the first and second term of (\ref{20}) can respectively lower bounded as
\begin{equation}\label{21}
\begin{aligned}
&(\nabla f_i(\bm{x}_i^{(k)})-\nabla f_i(\bm{x}_i^{*}))^T(\bm{x}_i^{(k)}-\bm{x}_i^{*})\\
& \ge \sigma_{f,i}^2\|\bm{x}_i^{(k)}-\bm{x}^{*}\|_{\boldsymbol{A}_i^T\boldsymbol{A}_i}^2,
\end{aligned}
\end{equation}
\begin{equation}\label{22}
\lambda(\partial r(\bm{x}_i^{(k)})-\partial r(\bm{x}_i^*))^T(\bm{x}_i^{(k)}-\bm{x}_i^{*})\ge 0,
\end{equation}
in which $\sigma_{f,i}$ is the $\sigma$ that satisfies Assumption 2 for function $f_i(\cdot)$, and $\|\bm{z}\|_{\boldsymbol{A}}^2\triangleq\bm{z}^T\boldsymbol{A}\bm{z}$.

By substituting (\ref{21}) and (\ref{22}) into (\ref{20}) and summing over $i=1,...,N$, we obtain
\begin{equation}\label{23}
\begin{aligned}
&\mathbb{E}[\|\bm{x}^{(k)}-\bm{x}^{*}\|_{\boldsymbol{M}}^2\\
&+(\bm{x}^{(k)}-\bm{x}^{(k-1)})^T(\boldsymbol{D}_{\eta}-2c\boldsymbol{D}\otimes\boldsymbol{I}_K)(\bm{x}^{(k)}-\bm{x}^{*})\\
&+\sum_{i=1}^N\sum_{j\in\mathcal{N}_i}(\bm{u}_{ij}^{(k+1)}-\bm{u}_{ij}^{*})^T(\bm{x}_i^{(k)}-\bm{x}^{*})\\
&+\sum_{i=1}^N\sum_{j\in\mathcal{N}_i}(\bm{v}_{ji}^{(k+1)}-\bm{v}_{ji}^{*})^T(\bm{x}_i^{(k)}-\bm{x}^{*})\\
&+2c\sum_{i=1}^N\sum_{j\in\mathcal{N}_i}(\frac{\bm{x}_i^{(k)}+\bm{x}_j^{(k)}}{2}-\frac{\bm{x}_i^{(k-1)}+\bm{x}_j^{(k-1)}}{2})^T(\bm{x}_i^{(k)}-\bm{x}_i^{*})]\\
&\le 0,
\end{aligned}
\end{equation}
where $\bm{x}^{(k)}=[(\bm{x}_1^{(k)})^T, ...,(\bm{x}_N^{(k)})^T]^T$, $\boldsymbol{D}_{\eta}=diag\{\frac{2}{\eta_{1k}}, ..., \frac{2}{\eta_{Nk}}\}\otimes\boldsymbol{I}_K$, $\boldsymbol{D}_{\sigma_f}=diag\{\sigma_1, ..., \sigma_N\}\otimes\boldsymbol{I}_K$, $\boldsymbol{M}=\tilde{\boldsymbol{A}}^T\boldsymbol{D}_{\sigma_f}\tilde{\boldsymbol{A}}$ and $\tilde{\boldsymbol{A}}=blkdiag\{\boldsymbol{A}_1, ...,\boldsymbol{A}_N\}$.
It can be observed from (\ref{13c}) and also the update of $\bm{p}$ that
\begin{equation}\label{24}
\bm{u}_{ij}^{*}+\bm{v}_{ij}^{*}=\bm{0},\quad \forall j, i,
\end{equation}
\begin{equation}\label{25}
\bm{u}_{ij}^{(k)}+\bm{v}_{ij}^{(k)}=\bm{0},\quad \forall j, i, k,
\end{equation}
given the initial $\bm{u}_{ij}^{(0)}+\bm{v}_{ij}^{(0)}=\bm{0}, \forall j, i$.
Based on the above properties, the third and the fourth term of (\ref{23}) can be written as
\begin{equation}\label{26}
\begin{aligned}
\mathbb{E}&[\sum_{i=1}^N\sum_{j\in\mathcal{N}_i}(\bm{u}_{ij}^{(k+1)}-\bm{u}_{ij}^{*})^T(\bm{x}_i^{(k)}-\bm{x}^{*})\\
&+\sum_{i=1}^N\sum_{j\in\mathcal{N}_i}(\bm{v}_{ji}^{(k+1)}-\bm{v}_{ji}^{*})^T(\bm{x}_i^{(k)}-\bm{x}^{*})]\\
&=\mathbb{E}[\sum_{i=1}^N\sum_{j\in\mathcal{N}_i}(\bm{u}_{ij}^{(k+1)}-\bm{u}_{ij}^{*})^T(\bm{x}_i^{(k)}-\bm{x}^{*})\\
&+\sum_{i=1}^N\sum_{j\in\mathcal{N}_i}(\bm{v}_{ij}^{(k+1)}-\bm{v}_{ij}^{*})^T(\bm{x}_j^{(k)}-\bm{x}^{*})]\\
&=\mathbb{E}[\sum_{i=1}^N\sum_{j\in\mathcal{N}_i}(\bm{u}_{ij}^{(k+1)}-\bm{u}_{ij}^{*})^T(\bm{x}_i^{(k)}-\bm{x}_j^{(k)})]\\
&=\mathbb{E}[\frac{2}{c}\sum_{i=1}^N\sum_{j\in\mathcal{N}_i}(\bm{u}_{ij}^{(k+1)}-\bm{u}_{ij}^{*})^T(\bm{u}_{ij}^{(k+1)}-\bm{u}_{ij}^{(k)})]\\
&\triangleq \mathbb{E}[\frac{2}{c}(\bm{u}^{(k+1)}-\bm{u}^{*})^T(\bm{u}^{(k+1)}-\bm{u}^{(k)})],
\end{aligned}
\end{equation}
where the third equality is owing to the update rule of $\bm{u}_{ij}$ and $\bm{u}^{(k)}$ is a vector that stacks $\bm{u}_{ij}^{(k)}$ for all $j\in\mathcal{N}_i, i=1,...,N$.
The fifth term of {({\ref{23}})} can be rearranged as $c(\bm{x}^{(k)}-\bm{x}^{(k-1)})^T[(\boldsymbol{D}+\boldsymbol{W})\otimes \boldsymbol{I}_K](\bm{x}^{(k)}-\tilde{\bm{x}}^*)$. The detail can be found in (A.16) in \cite{Inexact}. By substituting this and (\ref{26}) into (\ref{23}), we obtain
\begin{equation}\label{27}
\begin{aligned}
{\mathbb{E}}&\Big[\|\bm{x}^{(k)}-\tilde{\bm{x}}^*\|_{\boldsymbol{M}}^2+(\bm{x}^{(k)}-\bm{x}^{(k-1)})^T\boldsymbol{H}(\bm{x}^{(k)}-\tilde{\bm{x}}^*)\\
&+\frac{2}{c}(\bm{u}^{(k+1)}-\bm{u}^{*})^T(\bm{u}^{(k+1)}-\bm{u}^{(k)})\Big]\le 0,
\end{aligned}
\end{equation}
where we define $\boldsymbol{H}\triangleq\boldsymbol{D}_{\eta}+c((\boldsymbol{W}-\boldsymbol{D})\otimes\boldsymbol{I}_M)$.

Under the assumption that $\frac{2}{\eta_k}>c\lambda_{max}$ where $\lambda_{max}$ is the largest eigenvalue of $(\boldsymbol{D}-\boldsymbol{W})$, we can get that $\boldsymbol{H}\succ0$. Note that
\begin{equation}\label{28}
\begin{aligned}
(\bm{b}^{(k)}-&\bm{b}^{(k-1)})^T\boldsymbol{G}(\bm{b}^{(k)}-\bm{b}^{*})=\frac{1}{2}\|\bm{b}^{(k)}-\bm{b}^{*}\|_{\boldsymbol{G}}^2\\
&+\frac{1}{2}\|\bm{b}^{(k)}-\bm{b}^{(k-1)}\|_{\boldsymbol{G}}^2-\frac{1}{2}\|\bm{b}^{(k-1)}-\bm{b}^{*}\|_{\boldsymbol{G}}^2
\end{aligned}
\end{equation}
for any sequence $\bm{b}^{(k)}$ and matrix $\boldsymbol{G}\succ0$. By applying (\ref{28}) to the second and third {terms} in {(\ref{27})}, we obtain that
\begin{equation}\label{29}
\begin{aligned}
&{\mathbb{E}}\Big[(\bm{x}^{(k)}-\tilde{\bm{x}}^*)^T[\boldsymbol{M}+\frac{1}{2}\boldsymbol{H}](\bm{x}^{(k)}-\tilde{\bm{x}}^*)+\frac{1}{c}\|\bm{u}^{(k+1)}-\bm{u}^{*}\|_2^2\Big]\\
&\le {\mathbb{E}}\Big[\frac{1}{2}(\bm{x}^{(k-1)}-\tilde{\bm{x}}^*)^T\boldsymbol{H}(\bm{x}^{(k-1)}-\tilde{\bm{x}}^*)\\
&+\frac{1}{c}\|\bm{u}^{(k)}-\bm{u}^{*}\|_2^2-\frac{1}{c}\|\bm{u}^{(k+1)}-\bm{u}^{(k)}\|_2^2\\
&-\frac{1}{2}(\bm{x}^{(k)}-\bm{x}^{(k-1)})^T\boldsymbol{H}(\bm{x}^{(k)}-\bm{x}^{(k-1)})\Big].
\end{aligned}
\end{equation}
Here $\boldsymbol{M}=\tilde{\boldsymbol{A}}^T\boldsymbol{D}_{\sigma_f}\tilde{\boldsymbol{A}}$, where $\boldsymbol{D}_{\sigma_f}$ is a diagonal matrix and its diagonal elements $\sigma_f>0$. So we have $\boldsymbol{M}\succ0$.
Then same as proved in \cite{Inexact}, we can conclude that $\mathbb{E}[\|\bm{x}^{(k)}-\tilde{\bm{x}^*}\|_{\boldsymbol{H}}^2+\frac{1}{c}\|\bm{u}^{(k+1)}-\bm{u}^{*}\|_2^2]$ converges to $0$ and thus the proof is complete.
The proof of the convergence rate is almost the same as the proof of Theorem \ref{theorem_convergence}(b) in \cite{Inexact}, except that the expressions of $\boldsymbol{M}$ and $\boldsymbol{H}$ are different. To prove the linear convergence rate in expectation, we need to prove that for some $\delta>0$,
\begin{equation}
\begin{aligned}
&{\mathbb{E}}\Big[(\|\bm{x}^{(k)}-\tilde{\bm{x}}^*\|^2_{\alpha \boldsymbol{M}+\frac{1}{2}\boldsymbol{H}}+\frac{1}{c}\|\bm{u}^{(k+1)}-\bm{u}^*\|_2^2\Big]\\
&\le {\mathbb{E}}\Big[\frac{1}{1+\delta}(\|\bm{x}^{(k-1)}-\tilde{\bm{x}}^*\|^2_{\alpha \boldsymbol{M}+\frac{1}{2}\boldsymbol{H}}+\frac{1}{c}\|\bm{u}^{(k)}-\bm{u}^*\|_2^2)\Big].
\end{aligned}
\end{equation}
In both algorithms, $\boldsymbol{M}\succ0$ and $\boldsymbol{H}\succ0$. Thus, the conditions (A.36) in \cite{Inexact} can be satisfied for some $\delta>0$, and the convergence rate in expectation is linear.
\end{proof}

Theorem \ref{theorem_convergence} gives the convergence analysis on expectation.
{Specifically, Theorem \ref{theorem_convergence}a) shows that the variables in all nodes $\bm{x}_1^{(k)},...,\bm{x}_N^{(k)}$ converge to the same optimal point $\bm{x}^*$ in expectation, which indicates the convergence in both accuracy and consensus error. Then, Theorem \ref{theorem_convergence}b) shows the linear convergence rate of SCCD-ADMM in expectation, given the feature matrix with full column rank.}

Because of the stochastic sampling process, the variance is brought into the result.
{In addition to the convergence property in expectation, we will show that the convergence of SCCD-ADMM is within the variance bound in the following Theorem \ref{theorem_distance}.}
In order to analyze the variance of the stochastic algorithm SCCD-ADMM, we consider the corresponding deterministic algorithm named as DSCCD-ADMM. It updates the primal and dual variables using the same rules as SCCD-ADMM except that each node communicates with all of the neighboring nodes, without the sampling technique.
{By removing the expectation operator of (\ref{20}), the equality can still hold for DSCCD-ADMM by definition. Since the analysis after (\ref{20}) is all based on the expectation, the convergence of DSCCD-ADMM can be easily proved.
Under the definition of DSCCD-ADMM and its convergence, we give Lemma \ref{punbiased} as follows to help prove Theorem \ref{theorem_distance}, which shows that the convergence property of SCCD-ADMM is within the variance bound from DSCCD-ADMM.}

\newtheorem{lemma}{Lemma}
\begin{lemma}\label{punbiased}
Let $\hat{\bm{p}}_i^{(k)}$ and $\hat{\bm{x}}_i^{(k)}$ denote the $\bm{p}$-update and $\bm{x}$-update results of DSCCD-ADMM in the $k$-th iteration, respectively, with the $(k-1)$-th iteration values $\bm{p}_i^{(k-1)}$ and $\bm{x}_i^{(k-1)}$ ``synchronized" with those of SCCD-ADMM. Then we have $\mathbb{E}[\bm{p}_i^{(k)}]=\mathbb{E}[\hat{\bm{p}}_i^{(k)}]$ and $\hat{\bm{p}}_i^{(k)}-\bm{p}_i^{(k)}=\nabla h_{i,\mathcal{N}_{i}}^{(k)}(\bm{x}_i^{(k-1)})-\nabla h_{i,\mathcal{N}_{\text{c},i}}^{(k)}(\bm{x}_i^{(k-1)})$.
\end{lemma}
\begin{proof}
According to the update rule of $\bm{p}_{i}^{(k)}$, it can be derived that
\begin{equation}\nonumber
\begin{aligned}
\mathbb{E}[\bm{p}_i^{(k)}] &=\mathbb{E}[\bm{p}_i^{(k-1)}] + \mathbb{E}\Big [\frac{1}{Num_i}\sum_{j\in\mathcal{N}_{\text{c},i}}\frac{\bm{x}_i^{(k-1)}-\bm{x}_j^{(k-1)}}{w_{ij}^{(k)}}\Big]\\
					    &=\bm{p}_i^{(k-1)}+\frac{1}{Num_i}\sum_{j\in\mathcal{N}_{\text{c},i}}\mathbb{E}\Big[\frac{\bm{x}_i^{(k-1)}-\bm{x}_j^{(k-1)}}{w_{ij}^{(k)}}\Big]\\
					    &=\bm{p}_i^{(k-1)}+\frac{1}{Num_i}\sum_{j\in\mathcal{N}_{\text{c},i}}\sum_{j'\in\mathcal{N}_{i}}w_{ij'}^{(k)}\Big[\frac{\bm{x}_i^{(k-1)}-\bm{x}_{j'}^{(k-1)}}{w_{ij'}^{(k)}}\Big]\\
					    &=\bm{p}_i^{(k-1)}+\frac{1}{Num_i}\times Num_i\sum_{j'\in\mathcal{N}_{i}}\Big[\bm{x}_i^{(k-1)}-\bm{x}_{j'}^{(k-1)}\Big]\\
					    &=\bm{p}_i^{(k-1)}+\sum_{j\in\mathcal{N}_{i}}(\bm{x}_i^{(k-1)}-\bm{x}_j^{(k-1)})=\hat{\bm{p}}_i^{(k)}.
\end{aligned}
\end{equation}
Besides, we have
\begin{equation}\nonumber
\begin{aligned}
& \hat{\bm{p}}_i^{(k)}-\bm{p}_i^{(k)}\\
&=c\sum_{j\in\mathcal{N}_{i}}(\bm{x}_i^{(k-1)}-\bm{x}_j^{(k-1)})-\frac{c}{Num_i}\sum_{j\in{\mathcal{N}_{\text{c},i}}}\frac{(\bm{x}_i^{(k-1)}-\bm{x}_j^{(k-1)})}{w_{ij}^{(k)}}\\
&=\nabla h_{i,\mathcal{N}_{i}}^{(k)}(\bm{x}_i^{(k-1)})-\nabla h_{i,\mathcal{N}_{\text{c},i}}^{(k)}(\bm{x}_i^{(k-1)}).
\end{aligned}
\end{equation}
\end{proof}

\begin{theorem}\label{theorem_distance}
Let $\hat{\bm{x}}_i^{(k)}$ defined as in \textbf{Lemma \ref{punbiased}}. The distance between $\bm{x}_i^{(k)}$ and $\hat{\bm{x}}_i^{(k)}$, can be bounded as
$\mathbb{E}\Big[ \|\hat{\bm{x}}_i^{(k)}-\bm{x}_i^{(k)}\|_2^2 \Big ] \le
\dfrac{\eta_k^2}{2}\mathbb{V}\Big [\nabla h_{i,\mathcal{N}_{\text{c},i}}^{(k)}(\bm{x}_i^{(k-1)})\Big ]$, where the variance $\mathbb{V}\Big [\nabla h_{i,\mathcal{N}_{\text{c},i}}^{(k)}(\bm{x}_i^{(k-1)})\Big ] \triangleq \mathbb{E}\bigg[\Big\|\nabla h_{i,\mathcal{N}_{i}}^{(k)}(\bm{x}_i^{(k-1)})-\nabla h_{i,\mathcal{N}_{\text{c},i}}^{(k)}(\bm{x}_i^{(k-1)})\Big\|_2^2\bigg]$.
\end{theorem}
\begin{proof}
Considering the optimal condition of the update of $\bm{x}_i^{(k)}$ and $\hat{\bm{x}}_i^{(k)}$, we have
\begin{equation}\label{31}
\begin{aligned}
&\nabla f_i(\bm{x}_i^{(k)})+\lambda\partial r(\bm{x}_i^{(k)})+\frac{2}{\eta_k}(\bm{x}_i^{(k)}-\bm{x}_i^{(k-1)})\\
&+\bm{p}_i^{(k)}+\nabla h_{i,\mathcal{N}_{\text{c},i}}^{(k)}(\bm{x}_i^{(k-1)})=0,
\end{aligned}
\end{equation}

\begin{equation}\label{32}
\begin{aligned}
&\nabla f_i(\hat{\bm{x}}_i^{(k)})+\lambda\partial r(\hat{\bm{x}}_i^{(k)})+\frac{2}{\eta_k}(\hat{\bm{x}}_i^{(k)}-\bm{x}_i^{(k-1)})\\
&+\hat{\bm{p}}_i^{(k)}+\nabla h_{i,\mathcal{N}_{i}}^{(k)}(\bm{x}_i^{(k-1)})=0.
\end{aligned}
\end{equation}
Subtracting (\ref{31}) from (\ref{32}) and multiplying the both sides by $(\hat{\bm{x}}_i^{(k)}-\bm{x}_i^{(k)})$, we arrive at
\begin{equation}\label{33}
\begin{aligned}
&\Big (\nabla f_i(\hat{\bm{x}}_i^{(k)})-\nabla f_i(\bm{x}_i^{(k)})\Big)^T(\hat{\bm{x}}_i^{(k)}-\bm{x}_i^{(k)})+\\
&\lambda\Big (\partial r(\hat{\bm{x}}_i^{(k)})-\partial r(\bm{x}_i^{(k)})\Big )^T(\hat{\bm{x}}_i^{(k)}-\bm{x}_i^{(k)})+\\
&\frac{2}{\eta_k}(\hat{\bm{x}}_i^{(k)}-\bm{x}_i^{(k)})^T(\hat{\bm{x}}_i^{(k)}-\bm{x}_i^{(k)})+(\hat{\bm{p}}_i^{(k)}-\bm{p}_i^{(k)})^T(\hat{\bm{x}}_i^{(k)}-\bm{x}_i^{(k)})\\
&+\Big (\nabla h_{i,\mathcal{N}_{i}}^{(k)}(\bm{x}_i^{(k-1)})-\nabla h_{i,\mathcal{N}_{\text{c},i}}^{(k)}(\bm{x}_i^{(k-1)})\Big )^T(\hat{\bm{x}}_i^{(k)}-\bm{x}_i^{(k)})=0.
\end{aligned}
\end{equation}
Because of the convexity of $f_i(\bm{x})$ and $r(\bm{x})$, the first and second terms in (\ref{33}) can be respectively lower bounded as
\begin{equation}\label{34}
\Big (\nabla f_i(\hat{\bm{x}}_i^{(k)})-\nabla f_i(\bm{x}_i^{(k)})\Big)^T(\hat{\bm{x}}_i^{(k)}-\bm{x}_i^{(k)})\ge0,
\end{equation}
\begin{equation}\label{35}
\lambda\Big (\partial r(\hat{\bm{x}}_i^{(k)})-\partial r(\bm{x}_i^{(k)})\Big)^T(\hat{\bm{x}}_i^{(k)}-\bm{x}_i^{(k)})\ge0.
\end{equation}
According to (\ref{34})(\ref{35}) and combining {Lemma \ref{punbiased}}, we have
\begin{equation}
\begin{aligned}\label{variance}
&\mathbb{E}\Big[ \|\hat{\bm{x}}_i^{(k)}-\bm{x}_i^{(k)}\|_2^2\Big ]\\
&\le -\frac{\eta_k}{2}\mathbb{E}\Big [<\nabla h_{i,\mathcal{N}_{i}}^{(k)}(\bm{x}_i^{(k-1)})-\nabla h_{i,\mathcal{N}_{\text{c},i}}^{(k)}(\bm{x}_i^{(k-1)})\\
&+\hat{\bm{p}}_i^{(k)}-\bm{p}_i^{(k)}, \hat{\bm{x}}_i^{(k)}-\bm{x}_i^{(k)}>\Big ]\\
&\le \eta_k\mathbb{E}\bigg[ \Big\|\nabla h_{i,\mathcal{N}_{i}}^{(k)}(\bm{x}_i^{(k-1)})-\nabla h_{i,\mathcal{N}_{\text{c},i}}^{(k)}(\bm{x}_i^{(k-1)})\Big\|\Big\|\hat{\bm{x}}_i^{(k)}-\bm{x}_i^{(k)}\Big\|\bigg]\\
&\le \frac{\eta_k^2}{2}\Big \|\nabla h_{i,\mathcal{N}_{i}}^{(k)}(\bm{x}_i^{(k-1)})-\nabla h_{i,\mathcal{N}_{\text{c},i}}^{(k)}(\bm{x}_i^{(k-1)})\Big\|_2^2\\
&=\frac{\eta_k^2}{2}\mathbb{V}\Big[\nabla h_{i,\mathcal{N}_{\text{c},i}}^{(k)}(\bm{x}_i^{(k-1)})\Big],
\end{aligned}
\end{equation}
where the operation $<\bm{a},\bm{b}>$ denotes the inner product of  $\bm{a}$ and $\bm{b}$.
The second inequality is due to the Cauchy-Schwartz inequality and Lemma 1. The third inequality is due to Lemma 3 proven in \cite{Imsam} and here the function associated with the Bregman divergence is $2$-strongly convex.
\end{proof}

{Theorem \ref{theorem_distance} gives the bound of the distance $\mathbb{E}\Big[ \|\hat{\bm{x}}_i^{(k)}-\bm{x}_i^{(k)}\|_2^2 \Big ]$ w.r.t. the variance of $\nabla h_{i,\mathcal{N}_{\text{c},i}}^{(k)}$, which corresponds to the randomness brought by choosing the communication nodes.
As $k\to\infty$, the distance between $\bm{x}_i^{(k)}$ and $\hat{\bm{x}}_i^{(k)}$ converges to $0$ with $\eta_k\to0$, which indicates that the convergence holds in the limit. Theorem \ref{theorem_distance}, together with Theorem \ref{theorem_convergence}, enhances the convergence properties of SCCD-ADMM.}

According to the variance bound given in Theorem \ref{theorem_distance}, we consider reducing this variance as much as possible referring to the similar analysis technique for importance sampling{.}
The probability that node $i$ selects neighboring node $j$ in the $k$-th iteration is denoted as $w_{ij}^{(k)}$, which is proportional to the importance of node $j$.

\begin{theorem}\label{theorem_distribution}
Choosing the communication nodes is seen as a sampling process as mentioned before, and the best distribution of sampling is
\begin{equation*}
w_{ij}^{(k)}=\frac{\|\bm{x}_i^{(k-1)}-\bm{x}_j^{(k-1)}\|_2}{\sum_{l\in\mathcal{N}_{i}}\|\bm{x}_i^{(k-1)}-\bm{x}_j^{(k-1)}\|_2}.
\end{equation*}
\end{theorem}
\begin{proof}
According to the result in Theorem \ref{theorem_distance}, in order to reduce the objective as much as possible, we should choose $w_{ij}^{(k)}$ as the solution of the following optimization
\begin{equation}\nonumber
\mathop{\min}_{\textbf{w}_i^{(k)},w_{ij}^{(k)}\in[0,1],\sum_{j\in\mathcal{N}_{i}}w_{ij}^{(k)}=1}\mathbb{V}\Big[\nabla h_{i,\mathcal{N}_{\text{c},i}}^{(k)}(\bm{x}_i^{(k-1)})\Big].
\end{equation}
Then we expand the optimization problem and have
\begin{equation}\label{37}
\begin{aligned}
&\mathbb{V}\Big[\nabla h_{i,\mathcal{N}_{\text{c},i}}^{(k)}(\bm{x}_i^{(k-1)})\Big]\\
&=\mathbb{E}\Big\|\nabla h_{i,\mathcal{N}_{i}}^{(k)}(\bm{x}_i^{(k-1)})-\nabla h_{i,\mathcal{N}_{\text{c},i}}^{(k)}(\bm{x}_i^{(k-1)})\Big\|_2^2\\
&=\mathbb{E}\Big\|\frac{c}{w_{ij}^{(k)}}(\bm{x}_i^{(k-1)}-\bm{x}_j^{(k-1)})-c\sum_{m\in\mathcal{N}_{i}}(\bm{x}_i^{(k-1)}-\bm{x}_m^{(k-1)})\Big\|_2^2.\\
\end{aligned}
\end{equation}
{To simplify the expression of $\mathbb{V}\Big[\nabla h_{i,\mathcal{N}_{\text{c},i}}^{(k)}(\bm{x}_i^{(k-1)})\Big]$, we have
\begin{equation}\label{38}
\begin{aligned}
&\mathbb{V}\Big[\nabla h_{i,\mathcal{N}_{\text{c},i}}^{(k)}(\bm{x}_i^{(k-1)})\Big]\\
=&\mathbb{E}\Big\|\frac{c}{w_{ij}^{(k)}}(\bm{x}_i^{(k-1)}-\bm{x}_j^{(k-1)})-c\sum_{m\in\mathcal{N}_{i}}(\bm{x}_i^{(k-1)}-\bm{x}_m^{(k-1)})\Big\|_2^2\\
=&c\mathbb{E}\Big\|\frac{1}{w_{ij}^{(k)}}(\bm{x}_i^{(k-1)}-\bm{x}_j^{(k-1)})\Big\|_2^2\\
&-2c\mathbb{E}\Big\|\frac{1}{w_{ij}^{(k)}}(\bm{x}_i^{(k-1)}-\bm{x}_j^{(k-1)})\sum_{m\in\mathcal{N}_{i}}(\bm{x}_i^{(k-1)}-\bm{x}_m^{(k-1)})\Big\|\\
&+c\mathbb{E}\Big\|\sum_{m\in\mathcal{N}_{i}}(\bm{x}_i^{(k-1)}-\bm{x}_m^{(k-1)})\Big\|_2^2\\
=&c\sum_{j\in\mathcal{N}_{i}}w_{ij}^{(k)}\times\frac{1}{{w_{ij}^{(k)}}^2}\|\bm{x}_i^{(k-1)}-\bm{x}_j^{(k-1)}\|_2^2\\
&-2c\Big\|\sum_{j\in\mathcal{N}_{i}}(\bm{x}_i^{(k-1)}-\bm{x}_j^{(k-1)})\Big\|\Big\|\sum_{m\in\mathcal{N}_{i}}(\bm{x}_i^{(k-1)}-\bm{x}_m^{(k-1)})\Big\|\\
&+c\Big\|\sum_{m\in\mathcal{N}_{i}}(\bm{x}_i^{(k-1)}-\bm{x}_m^{(k-1)})\Big\|_2^2\\
=&c\sum_{j\in\mathcal{N}_{i}}\frac{1}{w_{ij}^{(k)}}\|\bm{x}_i^{(k-1)}-\bm{x}_j^{(k-1)}\|_2^2\\
&-c\Big\|\sum_{j\in\mathcal{N}_{i}}(\bm{x}_i^{(k-1)}-\bm{x}_j^{(k-1)})\Big\|_2^2.\\
\end{aligned}
\end{equation}
Since the second term in (\ref{38}) is fixed, minimizing $\mathbb{V}\Big[\nabla h_{i,\mathcal{N}_{\text{c},i}}^{(k)}(\bm{x}_i^{(k-1)})\Big]$ can be simplified as}
\begin{equation}
\min{\sum_{j\in\mathcal{N}_{i}}\frac{1}{w_{ij}^{(k)}}\|\bm{x}_i^{(k-1)}-\bm{x}_j^{(k-1)}\|_2^2}.
\end{equation}
According to the Cauchy-Schwarz inequality, it is easy to verify that the solution of the optimization problem is
\begin{equation}
w_{ij}^{(k)}=\frac{\|\bm{x}_i^{(k-1)}-\bm{x}_j^{(k-1)}\|_2}{\sum_{l\in\mathcal{N}_{i}}\|\bm{x}_i^{(k-1)}-\bm{x}_l^{(k-1)}\|_2}.
\end{equation}
\end{proof}
{By minimizing the variance, Theorem \ref{theorem_distribution} derives the sampling distribution for each node as given in (\ref{distribution1}).}
This result indicates that the neighboring node $j$ is more important for node $i$ if its update result $\bm{x}_j$ has a larger distance from $\bm{x}_i$. It is intuitive since a larger distance means that the difference in the data distribution between the two nodes is greater and thus node $i$ will get a more accurate result for the whole data space, with the help of the variables updated in node $j$.
As node $i$ chooses more communication nodes, it samples from its neighboring nodes and the variance becomes lower, which means a faster rate of convergence.
{The following Theorem \ref{theorem_variance} shows the distance bound given the sampling distribution in Theorem \ref{theorem_distribution}.}

\begin{theorem}\label{theorem_variance}
When node $i$ samples its communication nodes with the distribution $w_{ij}^{(k)}=\frac{\|\bm{x}_i^{(k-1)}-\bm{x}_j^{(k-1)}\|_2}{\sum_{l\in\mathcal{N}_{i}}\|\bm{x}_i^{(k-1)}-\bm{x}_l^{(k-1)}\|_2}$, the variance of the algorithm is bounded by
\begin{equation}
\begin{aligned}
\mathbb{E}\Big\|\hat{\bm{x}}_i^{(k)}-\bm{x}_i^{(k)}\Big\|_2^2\le\frac{c^2\eta_k^2}{2}&\Big[(\sum_{j\in\mathcal{N}_{i}}\|\bm{x}_i^{(k-1)}-\bm{x}_j^{(k-1)}\|_2)^2-\\
&\|\sum_{j\in\mathcal{N}_{i}}(\bm{x}_i^{(k-1)}-\bm{x}_j^{(k-1)})\|_2^2\Big].
\end{aligned}
\end{equation}
\end{theorem}
{\begin{proof}
Given the expression of $w_{ij}^{(k)}$, we substitute it into (\ref{38}) and we have
\begin{equation}
\begin{aligned}\label{39}
&\mathbb{V}\Big[\nabla h_{i,\mathcal{N}_{\text{c},i}}^{(k)}(\bm{x}_i^{(k-1)})\Big]\\
=&c\sum_{j\in\mathcal{N}_{i}}\frac{1}{w_{ij}^{(k)}}\|\bm{x}_i^{(k-1)}-\bm{x}_j^{(k-1)}\|_2^2\\
&-c\Big\|\sum_{j\in\mathcal{N}_{i}}(\bm{x}_i^{(k-1)}-\bm{x}_j^{(k-1)})\Big\|_2^2\\
=&c\sum_{j\in\mathcal{N}_{i}}\|\bm{x}_i^{(k-1)}-\bm{x}_j^{(k-1)}\|_2\sum_{l\in\mathcal{N}_{i}}\|\bm{x}_i^{(k-1)}-\bm{x}_l^{(k-1)}\|_2\\
&-c\Big\|\sum_{j\in\mathcal{N}_{i}}(\bm{x}_i^{(k-1)}-\bm{x}_j^{(k-1)})\Big\|_2^2\\
=&c\Big[(\sum_{j\in\mathcal{N}_{i}}\|\bm{x}_i^{(k-1)}-\bm{x}_j^{(k-1)}\|_2)^2-\|\sum_{j\in\mathcal{N}_{i}}(\bm{x}_i^{(k-1)}-\bm{x}_j^{(k-1)})\|_2^2\Big].\\
\end{aligned}
\end{equation}
Then according to Theorem \ref{theorem_distance}, we can get the variance bound under the given $w_{ij}^{(k)}$.
\end{proof}}

\section{Experiment}\label{experiment}
In this section, we evaluate the performance of the proposed SCCD-ADMM algorithm in terms of the convergence and the total cost of the system. {Additionally, we also show the impact of the network topology on the algorithms as well as the delay comparison.}

Logistic regression, as a classification problem using optimization to get the maximum probability result, is set as the optimization problem in the experiments. It is widely applied in machine learning area.
We define the objective function as:
\begin{equation}
 	\min_{\bm{x}_i} \sum_{i=1}^N\sum_{j=1}^{m}[\text{log}(1+e^{\bm{x}_i^T\boldsymbol{A}_{ij}})-\bm{b}_{ij}\bm{x}_i^T\boldsymbol{A}_{ij}]+\lambda r(\bm{x}_i).
\end{equation}
In the above objective function, each node has $m$ samples and the dimension of the feature vector is $M$. $(\boldsymbol{A}_i, \bm{b}_i)$ is the training data collected by node $i$, where $\boldsymbol{A}_i=[\boldsymbol{A}_{i1},...,\boldsymbol{A}_{im}]^T\in \mathds{R}^{M\times m}$ is the feature matrix and $\bm{b}_i=[b_{i1},...,{b}_{im}]^T\in\mathds{R}^m$ is the binary label. $\bm{x}_i$ is the variable to be optimized. In the following experiments, we set $m=20$ and $M=100$.

The network is generated using \textsl{networkx} packages in python. We use the \textsl{Erdos\_Renyi} random network model and set the parameter as $N=30$ and $p=0.5$, which means the network has $30$ nodes and each node has a probability of $0.5$ to connect with another node. The generated network topology is shown in Fig. \ref{fig2} and is used in the following experiments {in subsections A and B}.
The training data $\boldsymbol{A}_i$ and the `true' weight vector $\bm{x}_i^{\text{True}}$ are generated randomly and the labels $\bm{b}_i$ are generated using $\bm{b}_i=\text{sign}{(\boldsymbol{A}_i^T\bm{x}_i^{\text{True}}+v_i)}$, where $v_i$ is the noise vector and $v_i\sim\mathcal{N}(0,0.1)$.
The stopping criterions are $acc<0.1$ and $cserr<0.1$. $obj^{*}$ is calculated through centralized ADMM using all training data. We set $\lambda=0.4$.
\begin{figure} [!htp]
\vspace{-0.15 cm}
    \centering
       \includegraphics[width=0.9\linewidth]{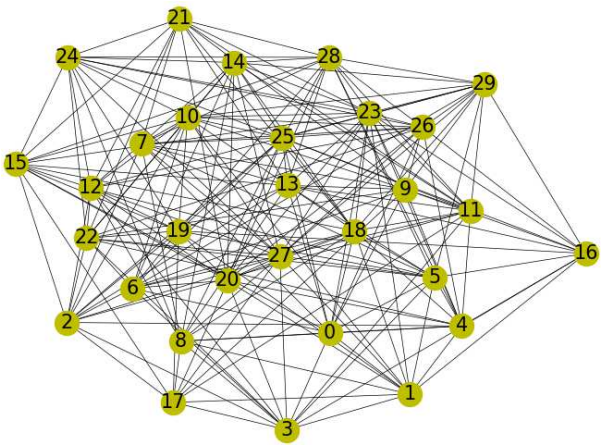}
	  \caption{Network topology}
	  \label{fig2}
\vspace{-0.4 cm}
\end{figure}

\subsection{Convergence Curve}\label{convergence}
In this part we apply the SCCD-ADMM algorithm to both $\ell_1$-regularized and $\ell_2$-regularized logistic regression problems.
\subsubsection{$\ell_2$-regularized objective}\label{l2-norm}
For $\ell_2$-regularized logistic regression, each feature vector and weight vector are generated to have 50 nonzero values and the training data $\boldsymbol{A}_i$ obeys a normal distribution $\mathcal{N}(0, 5)$.
 Define $co_{\text{rat}}=\dfrac{C_\text{cmm}}{C_\text{cmp}}$ which describes the numerical relationship between the communication cost and the computation cost. We examine the convergence of the SCCD-ADMM algorithm with different $co_{\text{rat}}$, i.e., $co_{\text{rat}}=0.1(c=0.3)$, $co_{\text{rat}}=0.6(c=0.3)$ and $co_{\text{rat}}=1.2(c=0.2)$. The step size of the searching procedure is set as $stepsize=2$. The traditional D-ADMM algorithm is used for comparison.
We set $\eta_k=\frac{D}{\sqrt{2k}}$ and $D=0.3$.
\begin{figure}[!htp]
\centering
\vspace{-0.25 cm}
\subfigure[Average of $Num_i$ of D-ADMM and SCCD-ADMM]{
\vspace{-0.5 cm}
	\label{a}
\vspace{-0.5 cm}
	\includegraphics[width=0.44\textwidth ,height=0.3\textwidth]{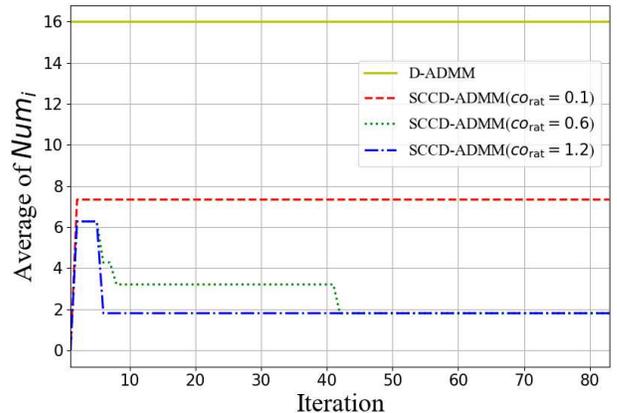}}
	
\vspace{-0.22 cm}
\subfigure[Accuracy Curve of D-ADMM and SCCD-ADMM]{
\vspace{-0.5 cm}
	\label{a}
\vspace{-0.5 cm}
	\includegraphics[width=0.44\textwidth ,height=0.3\textwidth]{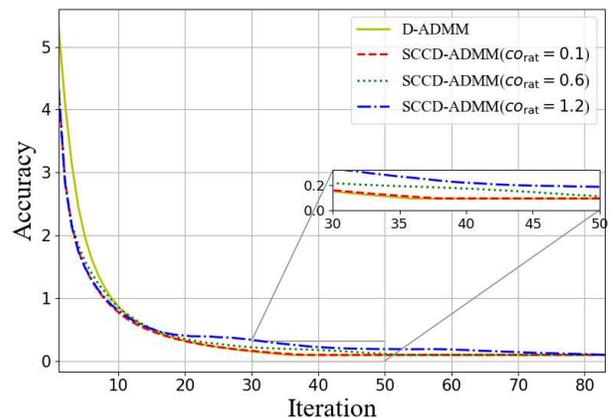}}

\vspace{-0.22 cm}
\subfigure[Consensus Error Curve of D-ADMM and SCCD-ADMM]{
	\vspace{-0.2 cm}
	\label{b}
	\vspace{-0.2 cm}
	\includegraphics[width=0.46\textwidth ,height=0.3\textwidth]{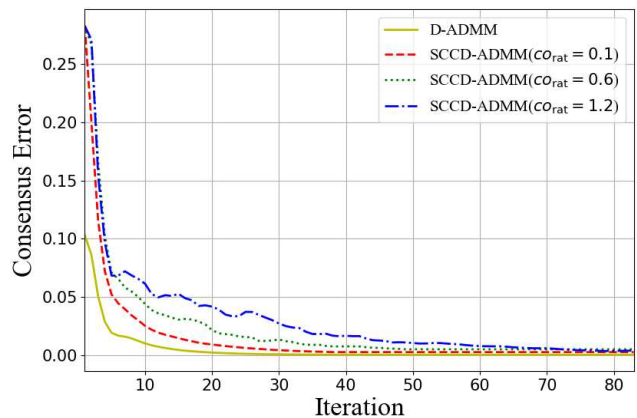}}
\caption{Average of $Num_i$ and Convergence of Different Algorithms for $\ell_2$-regularized Problem}
\vspace{-0.3 cm}
\label{fig3}
\vspace{-0.3 cm}
\end{figure}

The result is shown in Fig. 3. Fig. \ref{fig3}(a) shows how the average number of communication nodes of a node changes with the increase of iterations.
In the first iteration, each node only communicates with one of its neighbors as mentioned in \textbf{Algorithm \ref{searchnum}}. After that, $Num_i$ is determined by our proposed searching procedure.
It can be seen that with the increase of $co_{\text{rat}}$, the number of communication nodes reduces. It is because of the fact that higher $co_{\text{rat}}$ indicates that communication is more costly than computation, and hence according to our searching evaluation function, $Num_i$ is reduced to save communication cost. By contrast, the number of communication nodes in D-ADMM is always equal to that of the neighboring nodes, which is much larger than the number of the selected nodes in SCCD-ADMM.

Fig. \ref{fig3}(b) and (c) show the curves of the accuracy and the consensus error respectively. The consensus error gets converged before the accuracy, which suggests that it is the accuracy which mainly constrains the convergence in the case of  $\ell_2$-regularized objective.
The result shows that with $co_{\text{rat}}$ increasing, the convergence rate becomes slower, which is because high $co_{\text{rat}}$ leads to smaller $Num_i$ as shown in Fig. \ref{fig3}(a) and higher variance.
With less nodes communicating and more iterations to convergence, communication cost reduces at the price of higher computation cost, which shows the collaboration of communication and computation.
The figures of the cost will be shown in Section \ref{totalcost}.
Although SCCD-ADMM needs more iterations compared with D-ADMM, it can always converge and show great performance in terms of the cost.

\subsubsection{$\ell_1$-regularized objective}\label{l1nom}
Many large scale problems require the data to be sparse and thus $\ell_1$-regularized optimization is also the interest of research. In this part, we show the feasibility of our algorithm for $\ell_1$ norm.
Specifically, the local objective function for each node is $\phi_i(\bm{x})=f_i(\bm{x}_i)+0.4\|\bm{x}_i\|_1$. Each feature vector and the weight vector are generated to have 10 nonzero values to show the sparsity of the problem. The elements in the two vectors are generated randomly in the range of $[-1, 1]$.
We use FISTA \cite{FISTA} to solve the local optimization. Then for SCCD-ADMM, each node updates following
\begin{align*}
\tilde{\bm{x}}_i^{(l)}=\max\bigg\{-a, & \min\Big\{a, \mathcal{S}[\bm{z}_i^{(l-1)}-\rho_i^{(l)}\times\nabla h_{i, \mathcal{N}_{\text{c},i}}^{(k)}(\bm{z}_i^{(l-1)}) \\
					&+ 2\times \frac{\bm{z}_i^{(l-1)}-\bm{x}_i^{(k-1)}}{\eta_k}, \frac{\beta\rho_i^{(l)}}{N}]\Big\}\bigg\},
\end{align*}
\begin{align*}
\bm{z}_i^{(l-1)}=\tilde{\bm{x}}_i^{(l)}+\frac{l-1}{l+2}(\tilde{\bm{x}}_i^{(l)}-\tilde{\bm{x}}_i^{(l-1)}),
\end{align*}
where $l$ is the inner iteration number and $\mathcal{S}$ is the soft-threshold operator defined as $\mathcal{S}[m, n]\triangleq(m-n)_{+}+(-m-n)_{+}$ and $(x)_{+}\triangleq \max\{x, 0\}$. Each element of the variables is in the range of  $[-a,a]$.
Here, we set $\rho_i^{(l)}=0.01$, $\beta=0.1$ and $a=1$. The stopping criterion of the sub-optimization problem is $prg=\|\bm{z}_i^{(l-1)}-\tilde{\bm{x}}_i^{(l)}\|_2/(\rho_i^{(l)}\sqrt{K})<10^{-2}$. The optimal objective is calculated with all training data and $prg<10^{-6}$.

\begin{figure}[!htp]
\centering
\vspace{-0.3 cm}
\subfigure[Average of $Num_i$ of D-ADMM and SCCD-ADMM]{

\vspace{-0.5 cm}
\label{a}
	\vspace{-0.5 cm}
	\includegraphics[width=0.45\textwidth ,height=0.3\textwidth]{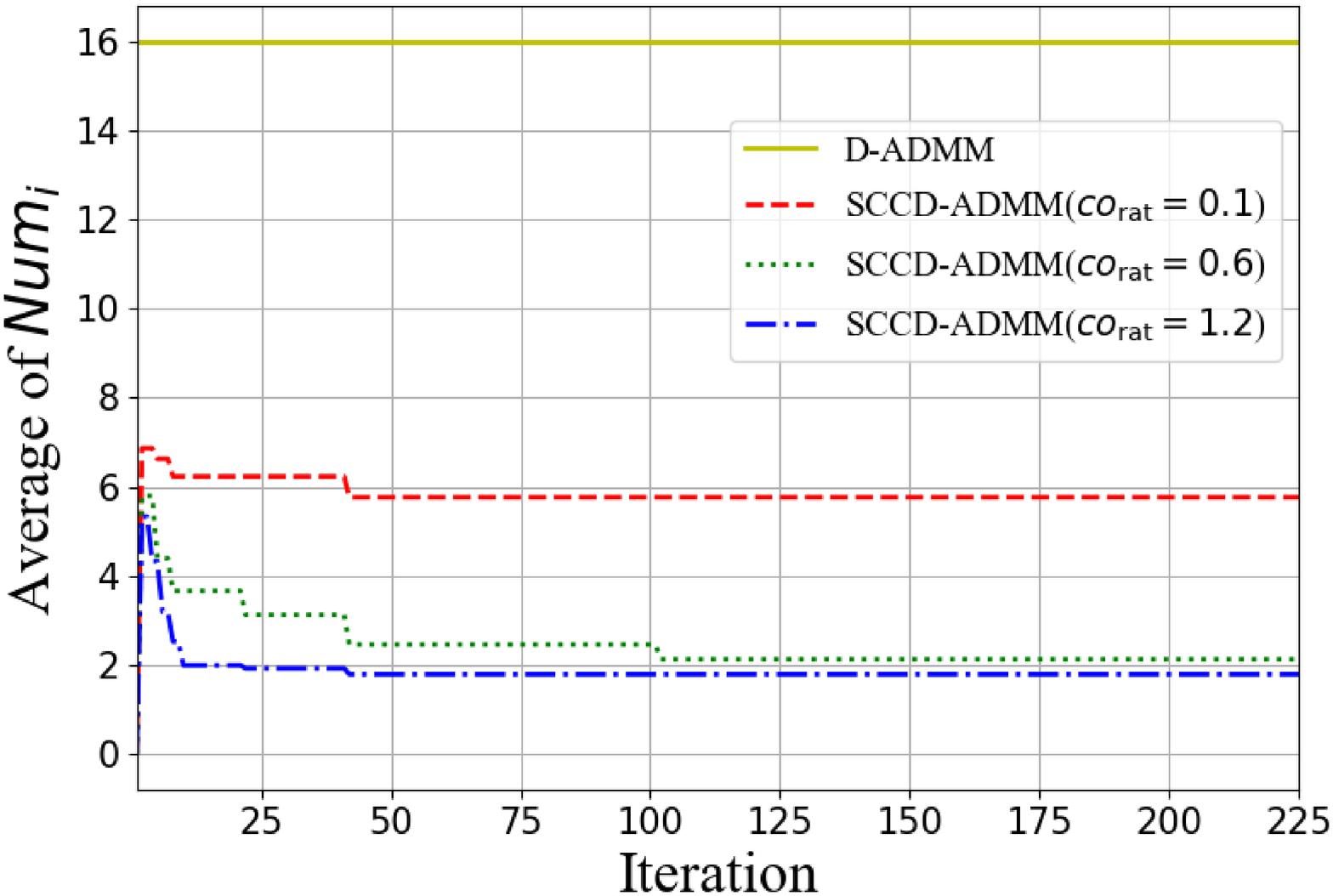}}

\vspace{-0.22 cm}
\subfigure[Accuracy Curve of D-ADMM and SCCD-ADMM]{

\vspace{-0.5 cm}
\label{a}
	\vspace{-0.5 cm}
	\includegraphics[width=0.45\textwidth ,height=0.3\textwidth]{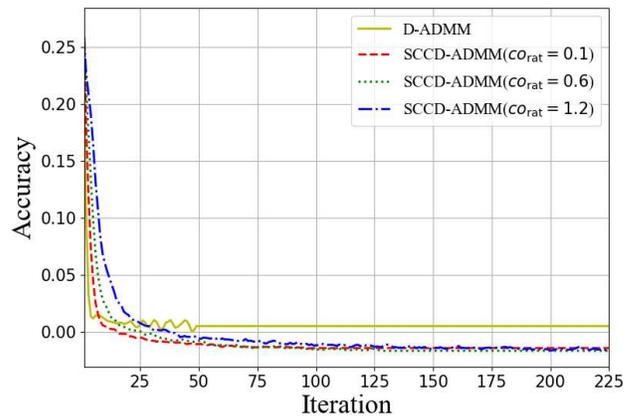}}

\vspace{-0.22 cm}
\subfigure[Consensus Error Curve of D-ADMM and SCCD-ADMM]{
	\vspace{-0.2 cm}
	\label{b}
	\vspace{-0.2 cm}
	\includegraphics[width=0.45\textwidth ,height=0.3\textwidth]{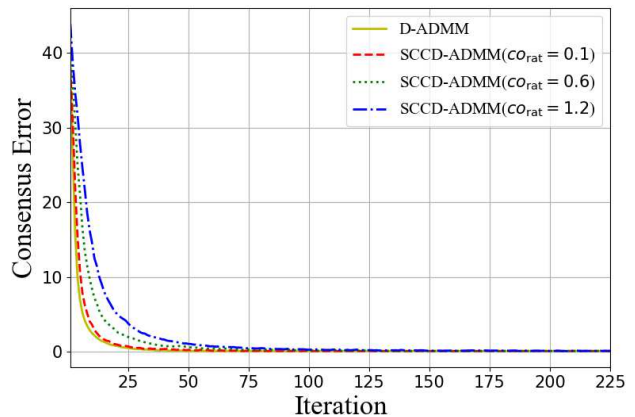}}

\caption{Average of $Num_i$ and Convergence of Different Algorithms for $\ell_1$-regularized problem}
\vspace{-0.3 cm}
\label{fig4}
\vspace{-0.3 cm}
\end{figure}

We implement the algorithm on different $co_{\text{rat}}$, i.e., $co_{\text{rat}}=0.1(c=0.007)$, $co_{\text{rat}}=0.6(c=0.007)$ and $co_{\text{rat}}=1.2(c=0.005)$. We set $stepsize=2$. The penalty parameter of D-ADMM is set as $c=0.006$.
The experiment result is shown in Fig. \ref{fig4}. Fig. \ref{fig4}(a) shows the same information as Fig. \ref{fig3}(a) that higher $co_{\text{rat}}$ results in smaller number of communication nodes.
As illustrated before, the main constraint on the convergence of $\ell_1$-regularized problem is consensus error and thus it can be seen in Fig. \ref{fig4}(b) and (c) that the accuracy reaches the stopping criterion ahead of the consensus error.
In addition, with the increase of $co_{\text{rat}}$, the number of communication nodes becomes smaller and thus convergence rate becomes slower. The experiment result validates that SCCD-ADMM can be applied to $\ell_1$-norm problem and also shows the tradeoff between communication and computation.

\subsection{Total Cost of the System}\label{totalcost}
In this subsection, we evaluate the computation and communication costs of the proposed SCCD-ADMM algorithm under different step sizes of the searching procedure and show its reduction of total cost compared with D-ADMM.
On the other hand, we also represent the performance of the tradeoff between the communication and computation costs, which aims to minimize the total cost of the system by assigning the proper number of communication nodes.


The objective function is set as $\sum_{i=1}^N\sum_{j=1}^{m}[\text{log}(1+e^{\bm{x}_i^T\boldsymbol{A}_{ij}})-\bm{b}_{ij}\bm{x}_i^T\boldsymbol{A}_{ij}]+0.4\|\bm{x}_i\|_2$. The training data is generated by the same way as in Section \ref{convergence}. The network topology is fixed and we average the result after simulating for 100 times.
{
For simplicity, we normalize $C_\text{cmp}$ as $1$. Since $C_\text{cmp}$ and $C_\text{cmm}$ are the same kind of cost measure with the same unit as defined in Section \ref{DecideNumber}, then $C_\text{cmm}$ equals to $co_{\text{rat}}$.}
We select the $co_{\text{rat}}$ ranging from $0$ to $2.4$ and in the extreme condition of $co_{\text{rat}}=0$, each node communicates with all of its neighboring nodes.
When $co_{\text{rat}}$ is higher, the communication cost dominates the total cost and naturally the total cost can be greatly saved.
We set three different step sizes for the searching procedure as $stepsize=1, 2, 3$.
The communication cost, computation cost and total cost is shown in Fig. \ref{fig5}.

\begin{figure}[!htp]
\centering
\vspace{-0.22 cm}
\subfigure[Total Communication Cost of the System]{
	\vspace{-0.5 cm}
	\label{b}
	\vspace{-0.0 cm}
	\includegraphics[width=0.44\textwidth ,height=0.3\textwidth]{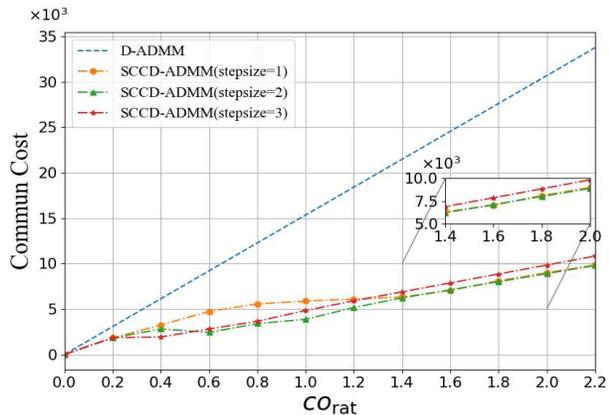}}

\vspace{-0.22 cm}
\subfigure[Total Computation Cost of the System]{
	\vspace{-0.5 cm}
	\label{b}
	\vspace{-0.0 cm}
	\includegraphics[width=0.44\textwidth ,height=0.3\textwidth]{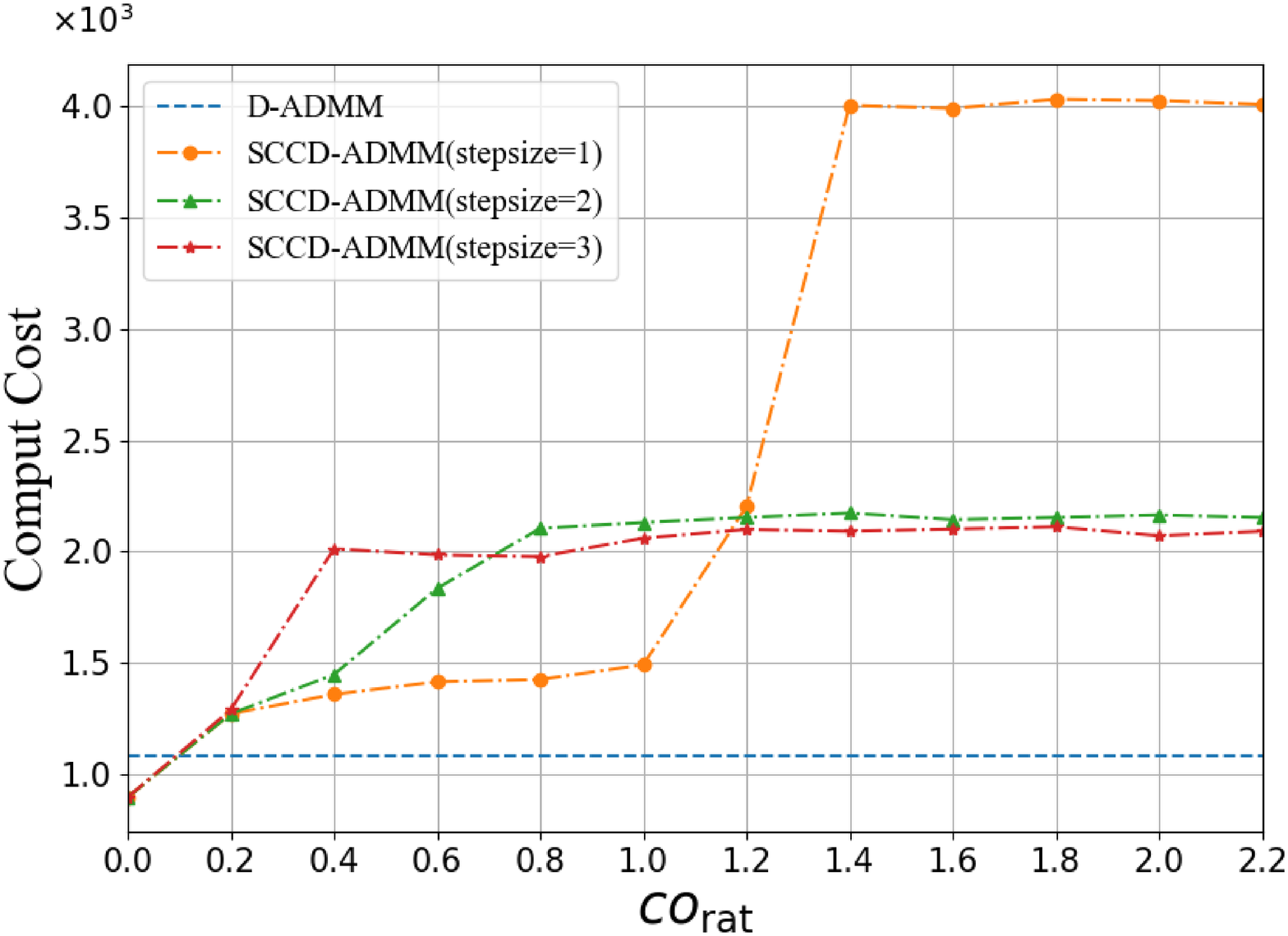}}

\vspace{-0.22 cm}
\subfigure[Total Cost of the System]{
\vspace{-0.5 cm}
	\label{a}
\vspace{-0.5 cm}
	\includegraphics[width=0.44\textwidth ,height=0.3\textwidth]{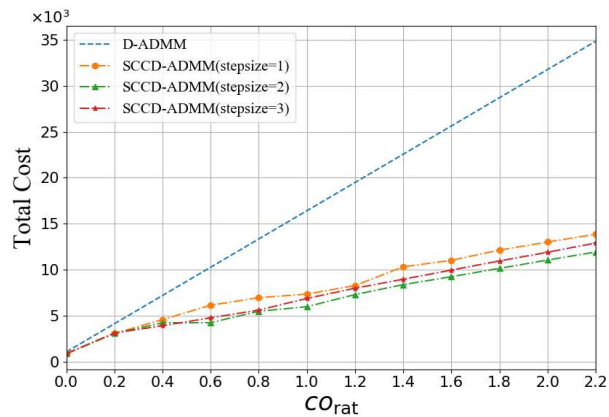}}
\caption{Cost Curves of $\ell_2$-regularized Problem}

\vspace{-0.3 cm}
\label{fig5}
\end{figure}

Fig. \ref{fig5}(a) shows the total communication cost of SCCD-ADMM is largely lessened compared with that of D-ADMM.
In D-ADMM, since the nodes have to communicate with all of their neighboring nodes in each iteration, the total communication cost increases linearly with $co_{\text{rat}}$.
However, in SCCD-ADMM, each node only communicates with a small subset of its neighboring nodes and the number of communication nodes continues decreasing with the increase of iterative times until it stabilizes.
Hence, the proposed SCCD-ADMM algorithm can significantly reduce the communication cost. When $co_{\text{rat}}$ is higher, more communication cost can be saved.

It can be seen from Fig. \ref{fig5}(b) that when $co_{\text{rat}}=0$, the computation cost of SCCD-ADMM is lower than that of D-ADMM, which indicates that in this extreme condition, each node communicates with all of its neighboring nodes and the algorithm can converge a little faster than D-ADMM.
As $co_{\text{rat}}$ increases, the computation cost of SCCD-ADMM is higher than that of D-ADMM. It is because SCCD-ADMM sacrifices the computation cost for smaller communication cost, which shows the tradeoff between communication and computation.

Fig. \ref{fig5}(b) exhibits that when $co_{\text{rat}}$ is higher than an inflection point, which we define as $\tilde{co}_{\text{rat}}$, the computation cost will reach a plateau. This is because when $co_{\text{rat}}$ is high enough, the number of communication nodes of each node reduces to its minimum value and thus the computation cost, which only depends on the iteration times and the searching steps, will not increase. From Fig. \ref{fig5}(b), we can observe that $\tilde{co}_{\text{rat}}$ decreases as $stepsize$ increases, which means that in the searching procedure, as the step size increases, the minimum value of $Num_i$ is obtained under a smaller $co_{\text{rat}}$. Moreover, when $co_{\text{rat}}>1.4$, as shown in Fig. \ref{fig5}(a) and (b), while the communication costs for the conditions of $stepsize=1$ and $stepsize=2$ are almost the same, the computation cost of $stepsize=1$ is much higher than that of $stepsize=2$. This is because that when $stepsize=1$, the algorithm needs more searching attempts towards the minimum $Num_i$ and thus consumes more computation cost.

Fig. \ref{fig5}(c) shows the total cost of the system, from which we can see the total cost can be substantially saved with less communication cost. As $co_{\text{rat}}$ increases, the total cost of D-ADMM increases faster than SCCD-ADMM since D-ADMM requires more communication nodes. The result in Fig. \ref{fig5}(c) reveals that our algorithm has greater advantages for large $co_{\text{rat}}$. In addition, we can see that different step sizes for searching procedure lead to different performance of the total cost and the condition of $stepsize=2$ performs the best. On the other hand, while the total costs of SCCD-ADMM vary for different step sizes, their values are all signally smaller than that of D-ADMM.

The above experiment results validate that SCCD-ADMM is more energy-efficient compared with traditional D-ADMM and the tradeoff is effective. With the increase of $co_{\text{rat}}$, SCCD-ADMM can save more energy by reducing the communication cost.

{
\subsection{Impact of Network Topology}\label{networktopology}
An important hyperparameter for distributed algorithms is the network topology, including the size of the network and its connectivity, whose impact will be given in this subsection.
To begin with, The number of nodes in the network has an impact on the convergence of the algorithms, including SCCD-ADMM and D-ADMM. We set the number as $N=10, 30, 60, 100$.
We use the same global data to guarantee the same centralized solution, i.e., the number of samples in each node is $m=60, 20, 10, 6$ respectively.
We use the \textsl{Erdos\_Renyi} random network model by fixing the connectivity probability as $p=0.5$ to generate the network. We set $co_\text{rat}=0.3$ and simulate for 50 times. The iteration number until convergence is given in Table. \ref{network_size}.
\begin{table}[!htp]
\centering
\caption{Iteration Numbers Under Different Network Size}
\label{network_size}
\begin{tabular}{|l|l|l|l|l|}
\hline
          & N=10 & N=30 & N=60 & N=100\\ \hline
SCCD-ADMM & 83   & 104  & 300  & 590\\ \hline
D-ADMM    & 23   & 36   & 128  & 310\\ \hline
\end{tabular}
\end{table}

As we can see in Table. \ref{network_size}, both D-ADMM and SCCD-ADMM has more iterations to convergence as the size of network becomes larger, which is because that each node has smaller size of data and needs more iterations to achieve global optimal point.

Next we give the comparison of the performance under different network connectivities.
Since the connection probability $p$ in \textsl{Erdos\_Renyi} model means the connection probability of two nodes and can proportionally reveal the connectivity of the network, we use $p$ to represent the connectivity of the network for simplicity and to show the tendency of the performance. We set $p=0.1$, $p=0.5$ and $p=0.9$ respectively and average after simulating for 100 times. In order to compare the performance under different connectivities more clearly, we fix $stepsize=2$ and plot the curves of different $p$ in Fig. \ref{connectivity}.

\begin{figure}[!htp]
\centering
\vspace{-0.22 cm}
\subfigure[Communication Cost Under Different $p$]{
	\vspace{-0.5 cm}
	\label{b}
	\vspace{-0.0 cm}
	\includegraphics[width=0.44\textwidth ,height=0.3\textwidth]{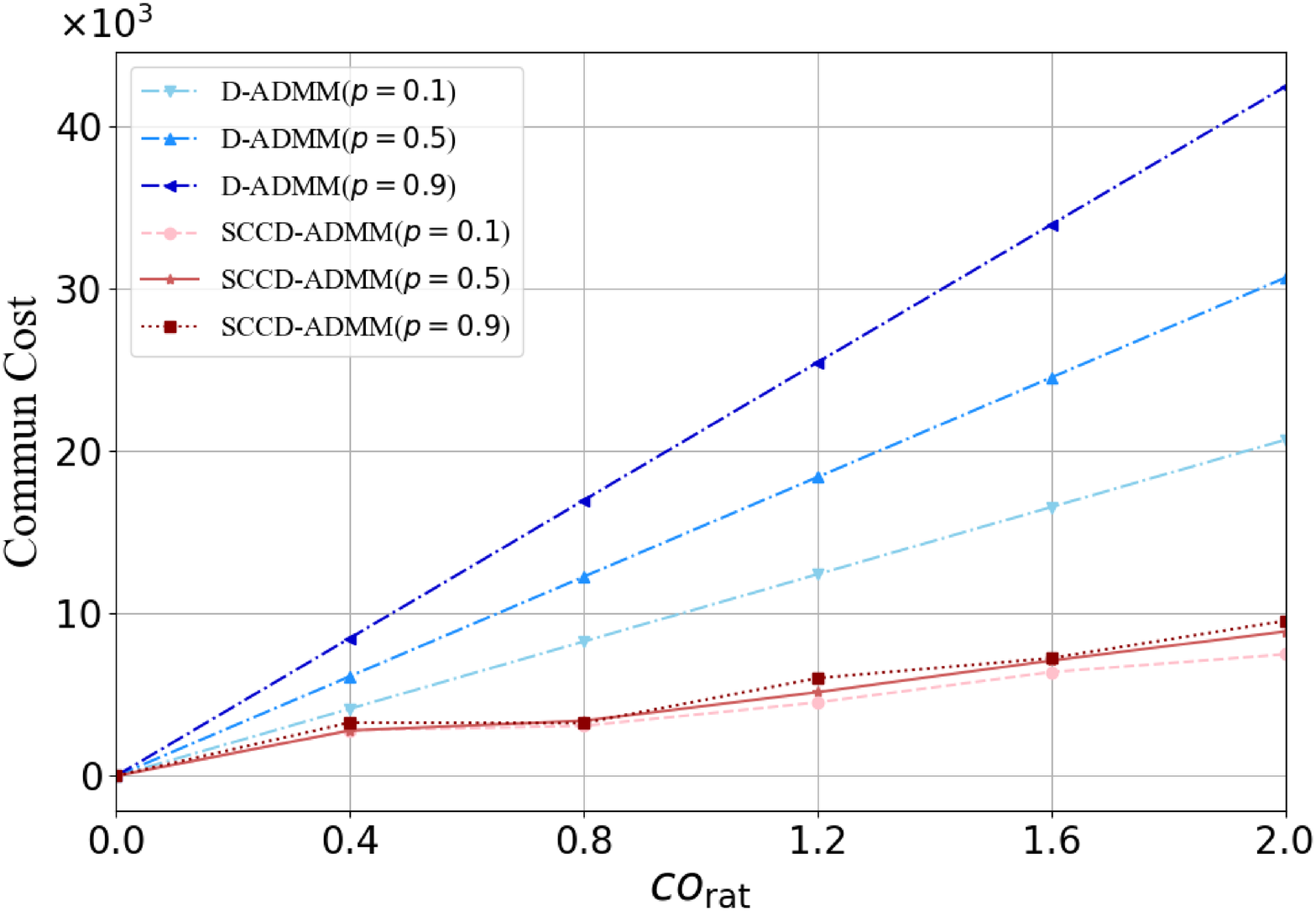}}

\vspace{-0.22 cm}
\subfigure[Computation Cost Under Different $p$]{
	\vspace{-0.5 cm}
	\label{b}
	\vspace{-0.0 cm}
	\includegraphics[width=0.44\textwidth ,height=0.3\textwidth]{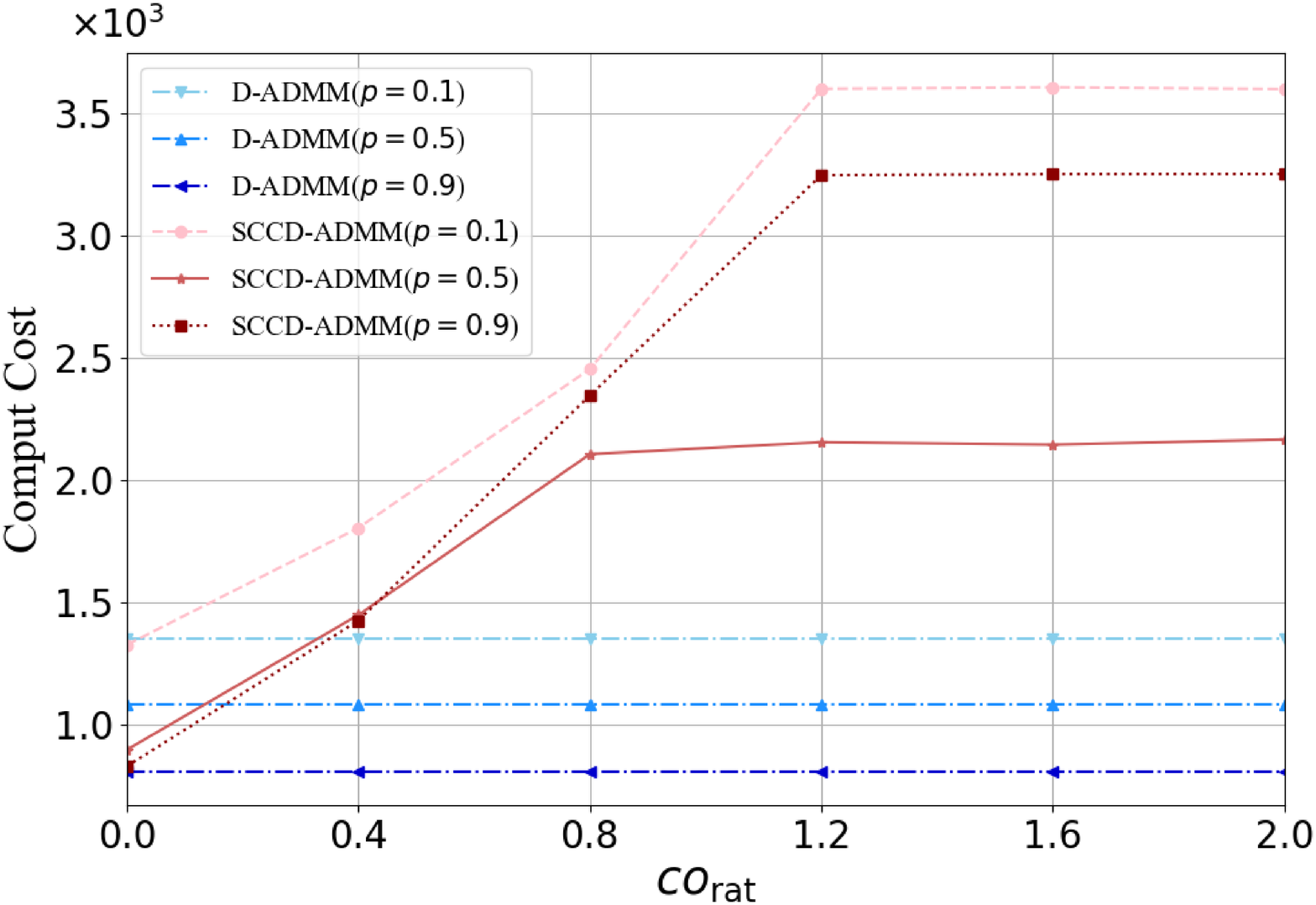}}

\vspace{-0.22 cm}
\subfigure[Total Cost Under Different $p$]{
\vspace{-0.5 cm}
	\label{a}
\vspace{-0.5 cm}
	\includegraphics[width=0.44\textwidth ,height=0.3\textwidth]{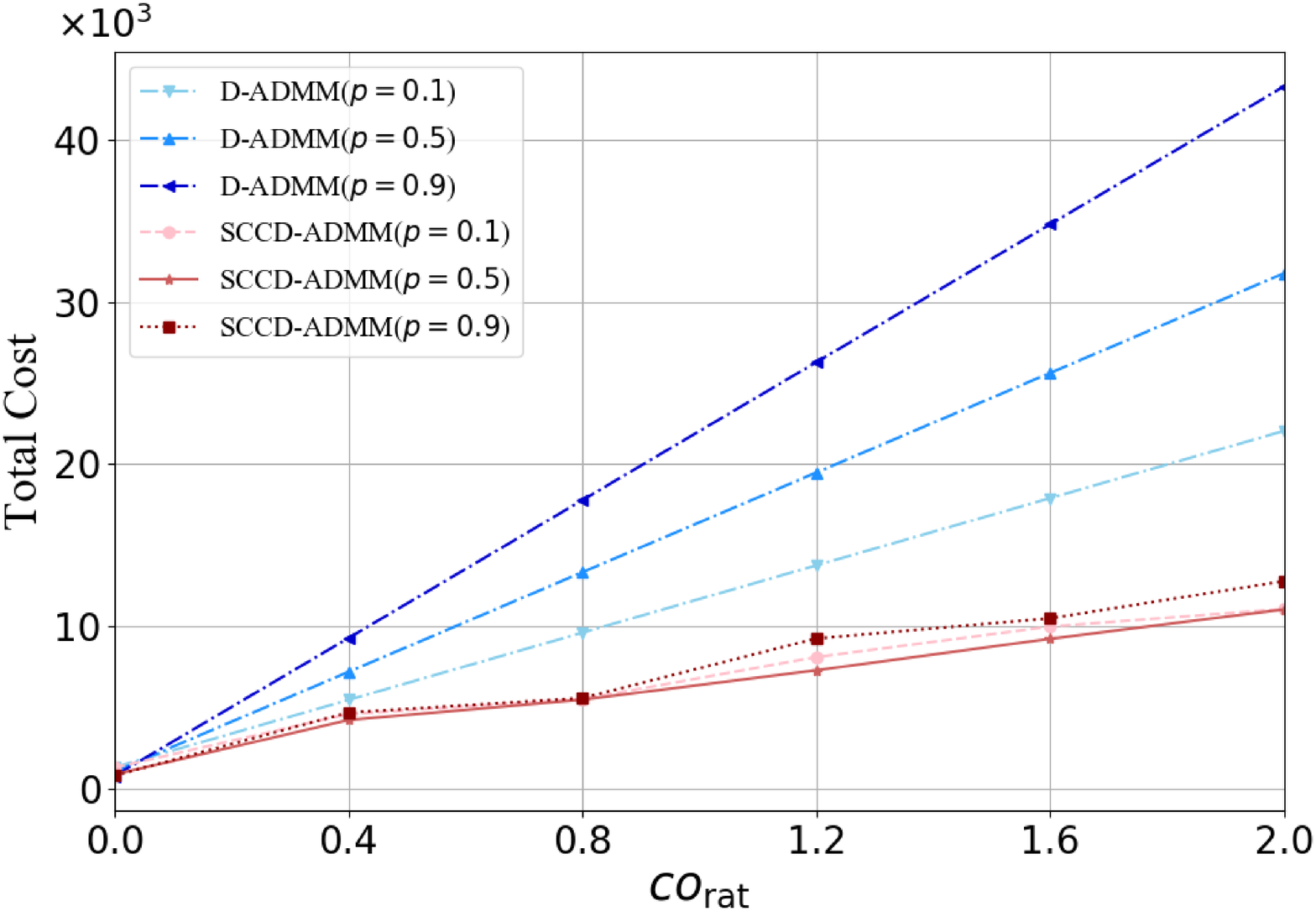}}
\caption{Cost Curves Under different $p$}

\vspace{-0.22 cm}
\label{connectivity}
\end{figure}

As shown in Fig. \ref{connectivity}(a), with the increasing connectivity of the graph, the communication cost of D-ADMM increases significantly due to a larger number of neighboring nodes, while the communication cost of SCCD-ADMM increases only slightly. This benefit comes from the searching procedure, which controls the number of communication nodes adaptively.
Fig. \ref{connectivity}(b) shows the computation cost comparison under different $p$. D-ADMM has less computation cost with larger $p$ because with more neighboring nodes to communicate with and more information received, the algorithm can converge faster.
However, SCCD-ADMM does not show the same phenomenon, because of the selective fewer communication nodes.
The computation cost of $p=0.5$ and $p=0.9$ is smaller than that of $p=0.1$, caused by smaller number of iterations.
This is because that the larger selection range of communication nodes makes it more likely to choose the better communication nodes w.r.t. the convergence rate of the algorithm.
On the other hand, the computation cost under $p=0.9$ is larger than that of $p=0.5$, resulted from additional steps for searching. One way to reduce the computation cost under larger $p$ is to consider a larger searching stepsize.
Fig. \ref{connectivity}(c) gives the overall cost, from which we can observe that SCCD-ADMM shows more benefit under larger $co_\text{rat}$ and larger $p$.
This is easy to understand since we aim to reduce the cost mainly by reducing communication links.
Thus SCCD-ADMM can save more cost compared with the traditional one, in the scenario where the unit communication cost is larger or the network is denser.

\subsection{Time Consumption Evaluation}\label{time}
In the practical implementation of distributed algorithms, the delay usually plays an important role.
Even though SCCD-ADMM focuses on the energy cost, we still need to give the performance evaluation and analysis of its delay comparison with the traditional D-ADMM.
In distributed ADMM algorithms, including SCCD-ADMM and D-ADMM, the delay comes from two aspects: the communication delay and the computation delay.

We consider the same setting as illustrated in Section \ref{convergence}1), with the same network topology and fixing $stepsize=2$.
For simplicity, we use the synchronous implementation. In this condition, the computation delay in each iteration is the maximum delay of updating process among all nodes and the communication delay in each iteration is the maximum transmission delay of all nodes. In addition, we assume that the transmission to one node from its communication nodes is conducted in a one-by-one way.
Then the communication delay of one node in one iteration is $\tau\times Num_i$, where $\tau$ is the transmission delay from one node to another and $Num_i$ is the number of communication nodes. Here we consider $\tau$ is equal among all transmission for simplicity.
The algorithms are implemented in Python 2.7 with Intel Core I5-9400F CPU with 2.9GHz. We consider two real communication systems with the communication rates of $54Mbps$ and $11Mbps$ under the IEEE standards \emph{IEEE 802.11g} \cite{IEEE 802.11g} and \emph{IEEE 802.11b} \cite{IEEE 802.11b} respectively.
The size of each transmission package is equal to the dimension of the variable $M=100$, where each data is a 32 bit float. Then $\tau=\frac{100\times32}{54\times10^{6}}=6.9\times 10^{-5}s$ for \emph{IEEE 802.11g} and $\tau=\frac{100\times32}{11\times10^{6}}=1.4\times 10^{-5}s$ for \emph{IEEE 802.11b}.
We compare D-ADMM and SCCD-ADMM in different $co_\text{rat}$ as shown in Table \ref{Time_table1} and Table \ref{Time_table2}.

\begin{table*}[!htp]
\centering
\caption{Delay comparison (in Seconds, \emph{IEEE 802.11g})}
\label{Time_table1}
\begin{tabular}{l|ll|ll|ll}
\hline
    & \multicolumn{2}{l|}{Communication Delay} & \multicolumn{2}{l|}{Computation Delay} & \multicolumn{2}{l}{Total Delay} \\ \hline
{$co_\text{rat}$}  & D-ADMM             & SCCD-ADMM           & D-ADMM            & SCCD-ADMM          & D-ADMM        & SCCD-ADMM       \\
0.1 & 0.04054            & 0.02455             & 0.20478           & 0.23588            & 0.24532       & 0.26043         \\
0.7 & 0.04054            & 0.01932             & 0.20478           & 0.28698            & 0.24532       & 0.30630         \\
1.4 & 0.04054            & 0.03097             & 0.20478           & 0.35800            & 0.24532       & 0.38997         \\
2.1 & 0.04054            & 0.03114             & 0.20478           & 0.35882            & 0.24532       & 0.38997         \\ \hline
\end{tabular}
\end{table*}

\begin{table*}[!htp]
\centering
\caption{Delay comparison (in Seconds, \emph{IEEE 802.11b})}
\label{Time_table2}
\begin{tabular}{l|ll|ll|ll}
\hline
    & \multicolumn{2}{l|}{Communication Delay} & \multicolumn{2}{l|}{Computation Delay} & \multicolumn{2}{l}{Total Delay} \\ \hline
{$co_\text{rat}$}  & D-ADMM             & SCCD-ADMM           & D-ADMM            & SCCD-ADMM          & D-ADMM        & SCCD-ADMM       \\
0.1 & 0.19900            & 0.12049             & 0.20478           & 0.23588            & 0.40378       & 0.35638         \\
0.7 & 0.19900            & 0.09485             & 0.20478           & 0.28698            & 0.40378       & 0.38183         \\
1.4 & 0.19900            & 0.15266             & 0.20478           & 0.35800            & 0.40378       & 0.51066         \\
2.1 & 0.19900            & 0.15287             & 0.20478           & 0.35882            & 0.40378       & 0.51170         \\ \hline
\end{tabular}
\end{table*}

As shown in both Table \ref{Time_table1} and \ref{Time_table2}, the communication time consumption of SCCD-ADMM is smaller than that of D-ADMM, resulted from a smaller number of communication nodes.
However, SCCD-ADMM has a relatively larger computation delay compared with that of D-ADMM. This comes from the additional searching process and a larger iteration number.
In Table \ref{Time_table1}, the total delay of SCCD-ADMM is larger than that of D-ADMM due to the high communication rate and the dominance of computation delay.
On the other hand, when the communication rate is low and the communication delay becomes dominant, SCCD-ADMM can still outperform D-ADMM in time delay under small $co_\text{rat}$ as shown in Table \ref{Time_table2}. This indicates the potential benefits of SCCD-ADMM w.r.t. time delay in the terrible communication environment.
In addition, as $co_\text{rat}$ increases, the computation delay of SCCD-ADMM increases as well because of a larger number of iterations resulted from fewer communication nodes. When $co_\text{rat}$ is large enough, the communication delay and computation delay reach a plateau and remain steady, which is also shown in computation cost curve as shown in Fig. \ref{fig5}(b).
}

\section{Conclusion}\label{conclusion}
In this paper, a novel distributed optimization algorithm called SCCD-ADMM algorithm is proposed to save the total cost of the system while implementing the distributed ADMM algorithm.
In the algorithm, each node adaptively determines the number of communication nodes following the given searching procedure, while the specific communication nodes are chosen according to the derived sampling distribution. After receiving the information from the selected neighboring nodes, each node updates its local information with the newly-designed update rule and its convergence analysis is given.
Compared with the traditional distributed ADMM, the proposed algorithm reduces the communication nodes and thus trading computation cost for less communication cost.
By making a favorable tradeoff between communication and computation costs, the total cost of the system is largely saved. Numerical experiments validate the superiority of our algorithms over the conventional one.

There are many future research topics. One of them is extending the algorithm to the conditions where the links between nodes have different costs of communication. It requires each node to choose communication nodes with a different criterion since in addition to the data importance, each node has to consider the various communication cost of its neighboring nodes. Another topic is changing the method of deciding the number of communication nodes. In this paper, the heuristic searching procedure is applied. Given the exact convergence rate corresponding to the number of communication nodes, an optimization problem can be derived to get the optimal number which can make the best tradeoff between communication and computation.

\begin{IEEEbiography}[{\includegraphics[width=1in,height=1.25in,clip,keepaspectratio] {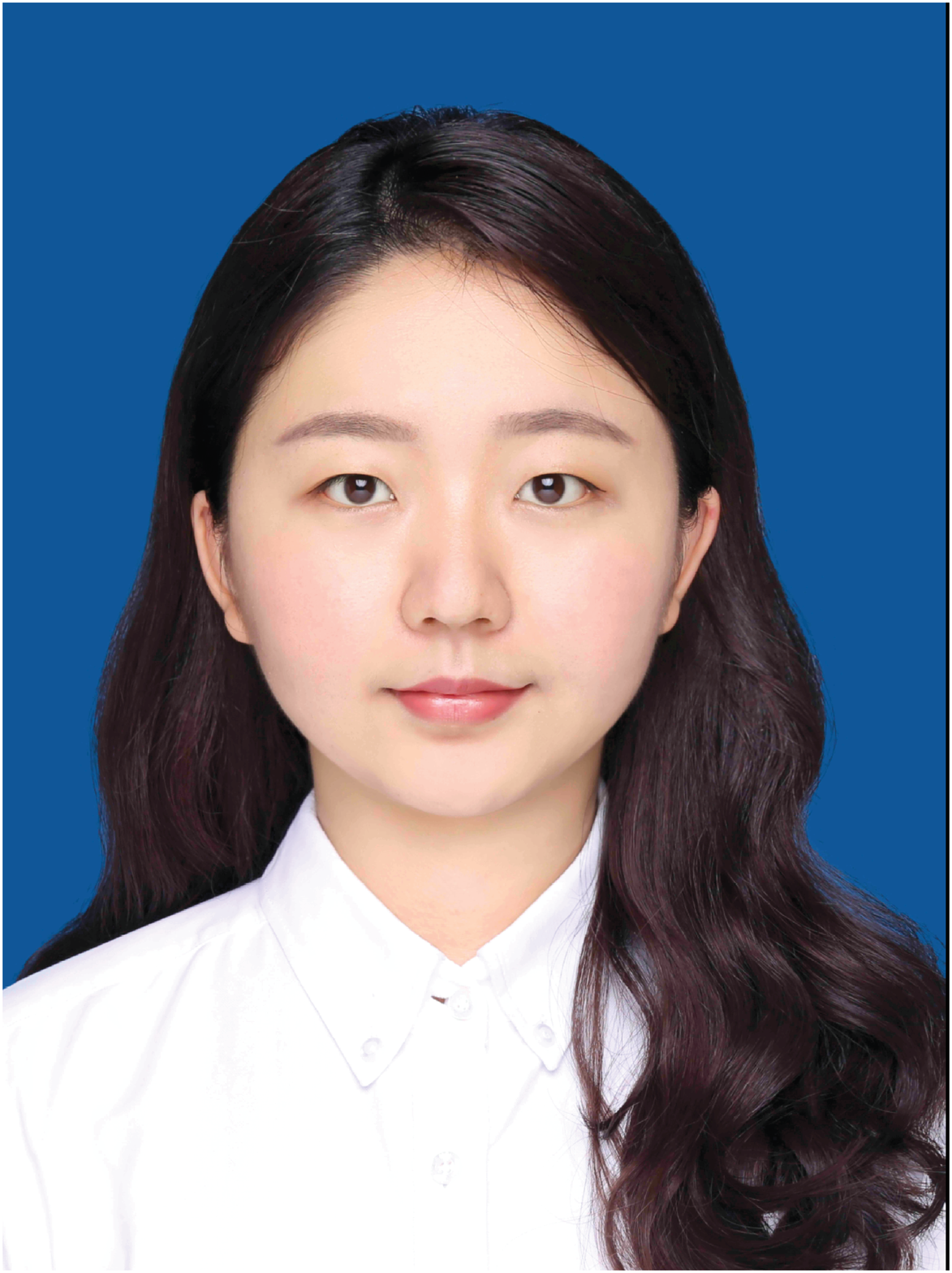}}]{Zhuojun~Tian} (S'19) received the B.S.Eng. degree in information engineering from Zhejiang Univerisy, Hangzhou, China, in 2019. She is currently pursuing the Ph.D. degree in information and communication engineering under the supervision of Prof. Z. Zhang at Zhejiang University. Her current research interests include distributed algorithms, massive MIMO, and machine learning for wireless networks.
\end{IEEEbiography}

\begin{IEEEbiography}[{\includegraphics[width=1in,height=1.25in,clip,keepaspectratio] {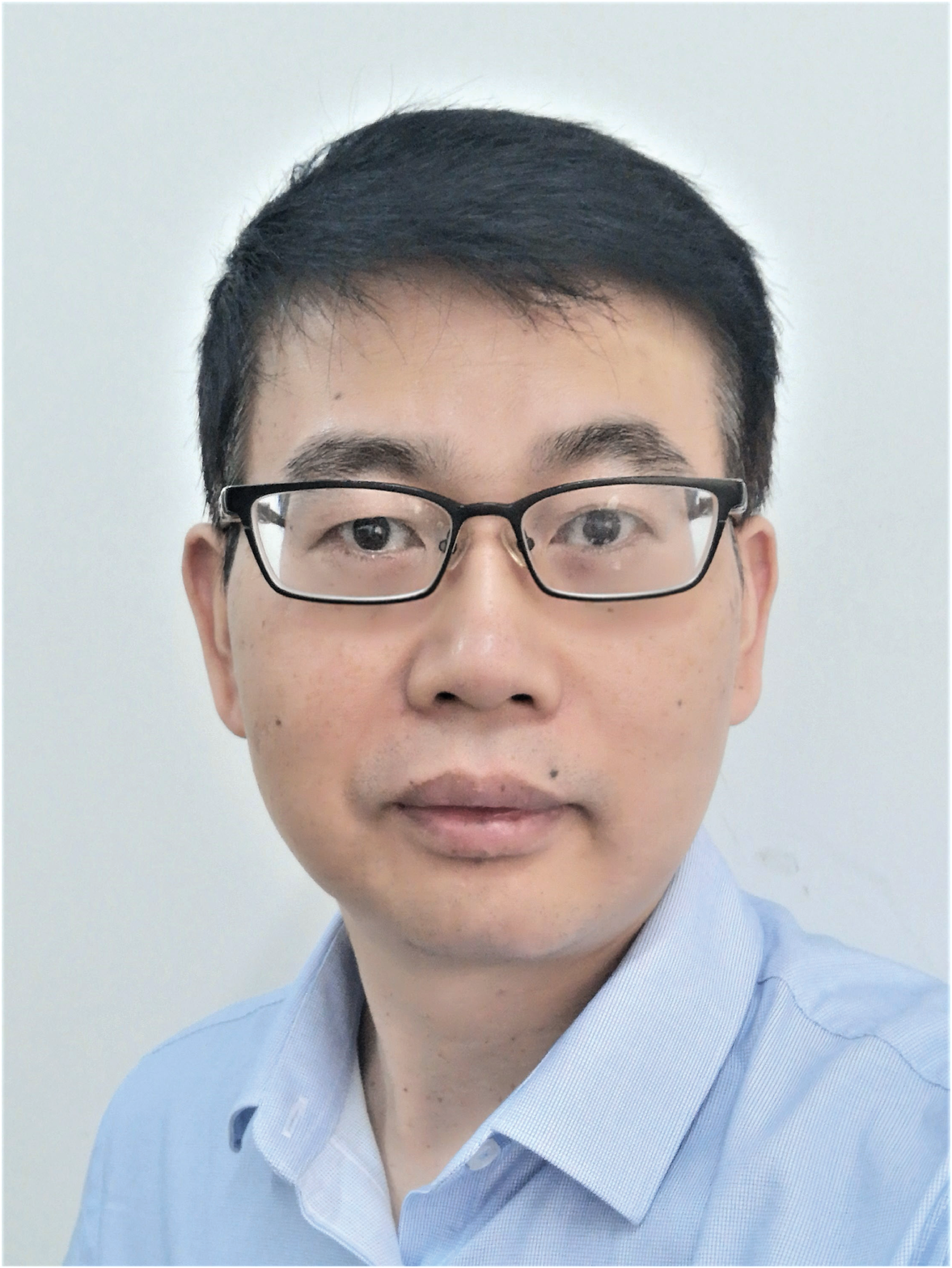}}]{Zhaoyang~Zhang} (M'02) received his Ph.D. degree from Zhejiang University, Hangzhou, China, in 1998, where he is currently a Qiushi Distinguished Professor. His current research interests are mainly focused on the fundamental aspects of wireless communications and networking, such as information theory and coding, network signal processing and distributed learning, AI-empowered communications and networking, network intelligence with synergetic sensing, computation and communication, etc. He has co-authored more than 300 peer-reviewed international journal and conference papers, and is a co-recipient of 7 conference best paper awards including ICC 2019. He was awarded the National Natural Science Fund for Distinguished Young Scholars by NSFC in 2017.

Dr. Zhang is serving or has served as Editor for \textsc{IEEE Transactions on Wireless Communications}, \textsc{IEEE Transactions on Communications} and \textsc{IET Communications}, etc, and as General Chair, TPC Co-Chair or Symposium Co-Chair for WCSP 2013/2018, Globecom 2014 Wireless Communications Symposium, and VTC-Spring 2017 Workshop HMWC, etc.
\end{IEEEbiography}

\begin{IEEEbiography}[{\includegraphics[width=1in,height=1.25in,clip,keepaspectratio] {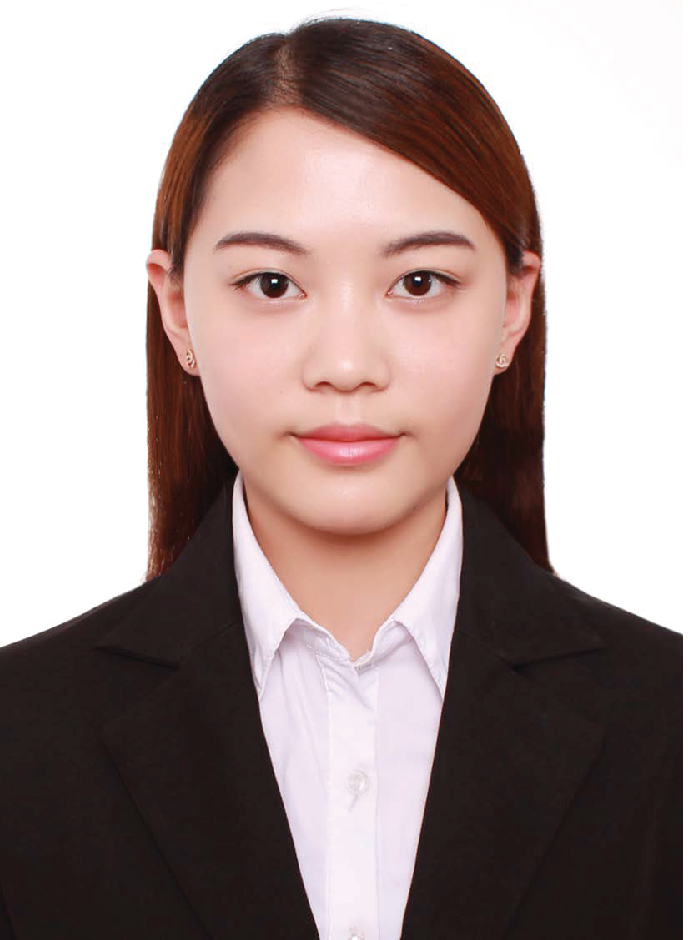}}]{Jue~Wang} (S'18) received the B.S.Eng. degree in communication engineering from the Department of Communication and Information Engineering, Nanjing University of Posts and Telecommunications, Nanjing, China, in 2016. She is currently pursuing the Ph.D. degree in information and communication engineering under the supervision of Prof. Z. Zhang at Zhejiang University. Her current research interests include signal processing, massive access, massive MIMO and machine learning.
\end{IEEEbiography}

\begin{IEEEbiography}[{\includegraphics[width=1in,height=1.25in,clip,keepaspectratio] {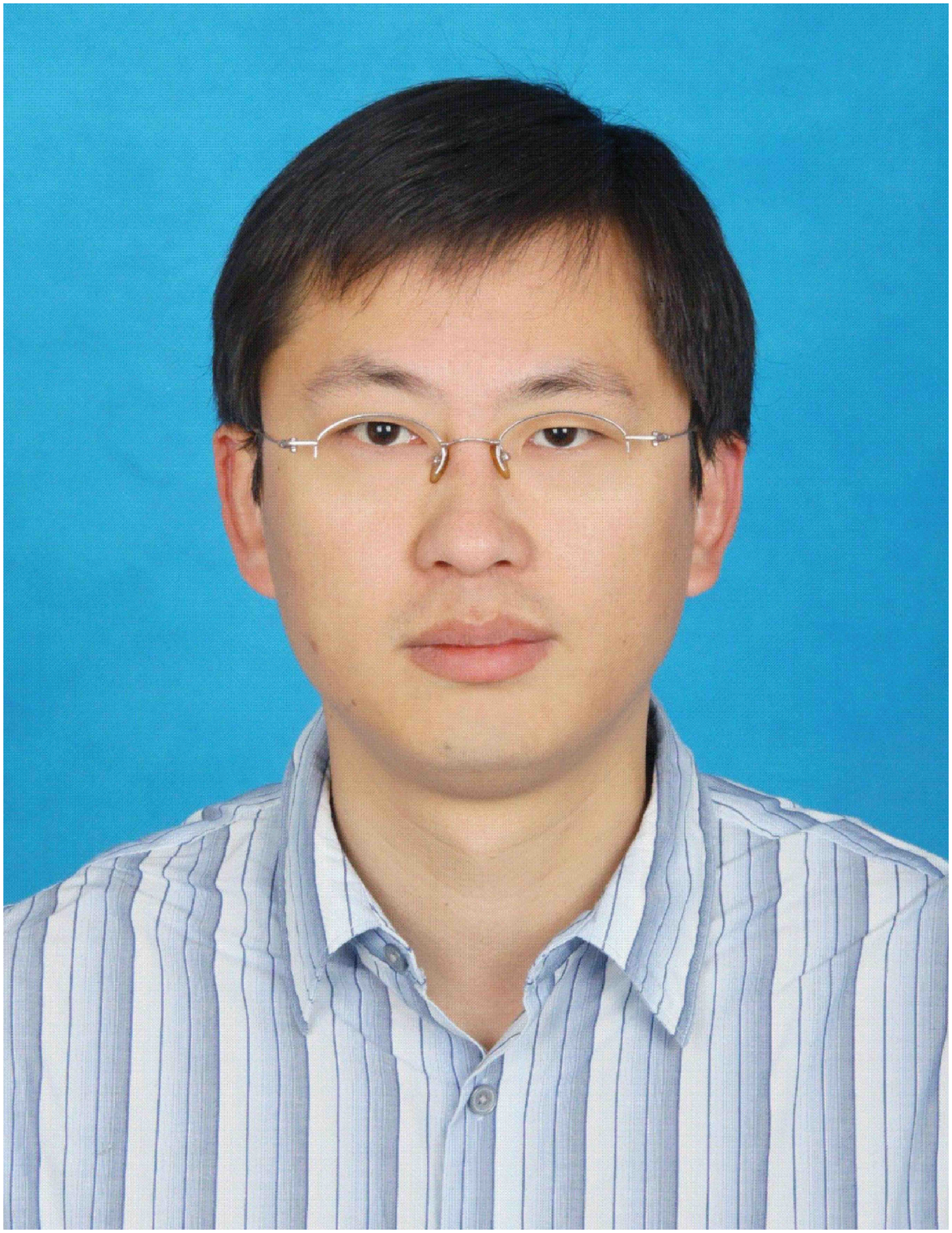}}]{Xiaoming~Chen} (M'10-SM'14) received the B.Sc. degree from Hohai University in 2005, the M.Sc. degree from Nanjing University of Science and Technology in 2007 and the Ph. D. degree from Zhejiang University in 2011, all in electronic engineering. He is currently a Professor with the College of Information Science and Electronic Engineering, Zhejiang University, Hangzhou, China. From March 2011 to October 2016, He was with Nanjing University of Aeronautics and Astronautics, Nanjing, China. From February 2015 to June 2016, he was a Humboldt Research Fellow at the Institute for Digital Communications, Friedrich-Alexander-University Erlangen-N\"urnberg (FAU), Germany. His research interests mainly focus on 5G/6G key techniques, Internet of Things, and smart communications.

Dr. Chen is currently serving as an Editor for the \textsc{IEEE Transactions on Communications} and the \textsc{IEEE Communications Letters}, and a Guest Editor for the \textsc{IEEE Journal on Selected Areas in Communications} ``Massive Access for 5G and Beyond" and the \textsc{IEEE Wireless Communications} ``Massive Machine-Type Communications for IoT". He received the Best Paper Awards at the IEEE International Conference on Communications (ICC) 2019, and the IEEE/CIC International Conference on Communications in China (ICCC) 2018.
\end{IEEEbiography}

\begin{IEEEbiography}[{\includegraphics[width=1in,height=1.25in,clip,keepaspectratio] {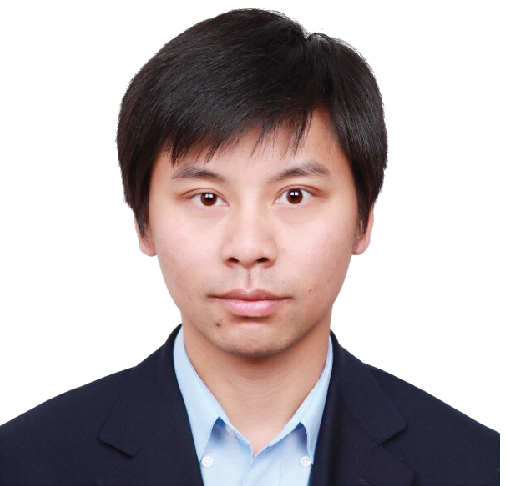}}]{Wei~Wang} (S'08-M'10-SM'15) received the B.S. and Ph.D. degrees from the Beijing University of Posts and Telecommunications, China, in 2004 and 2009, respectively. From 2007 to 2008, he was a Visiting Student with the University of Michigan, Ann Arbor, USA. From 2013 to 2015, he was a Hong Kong Scholar with the Hong Kong University of Science and Technology, Hong Kong. He is currently a Professor with the College of Information Science and Electronic Engineering, Zhejiang University, China. His research interests mainly focus on low-latency wireless communications, mobile edge computing, and stochastic optimization for wireless networks. He is the Editor of the book entitled Cognitive Radio Systems, and serves as an editor of Series on \textsc{Network Softwarization \& Enablers}, \textsc{IEEE Journal of Selected Areas in Communications}, \textsc{IEEE Access}, \textsc{Transactions on Emerging Telecommunications Technologies}, and \textsc{KSII Transactions on Internet and Information Systems}.
\end{IEEEbiography}

\begin{IEEEbiography}[{\includegraphics[width=1in,height=1.25in,clip,keepaspectratio] {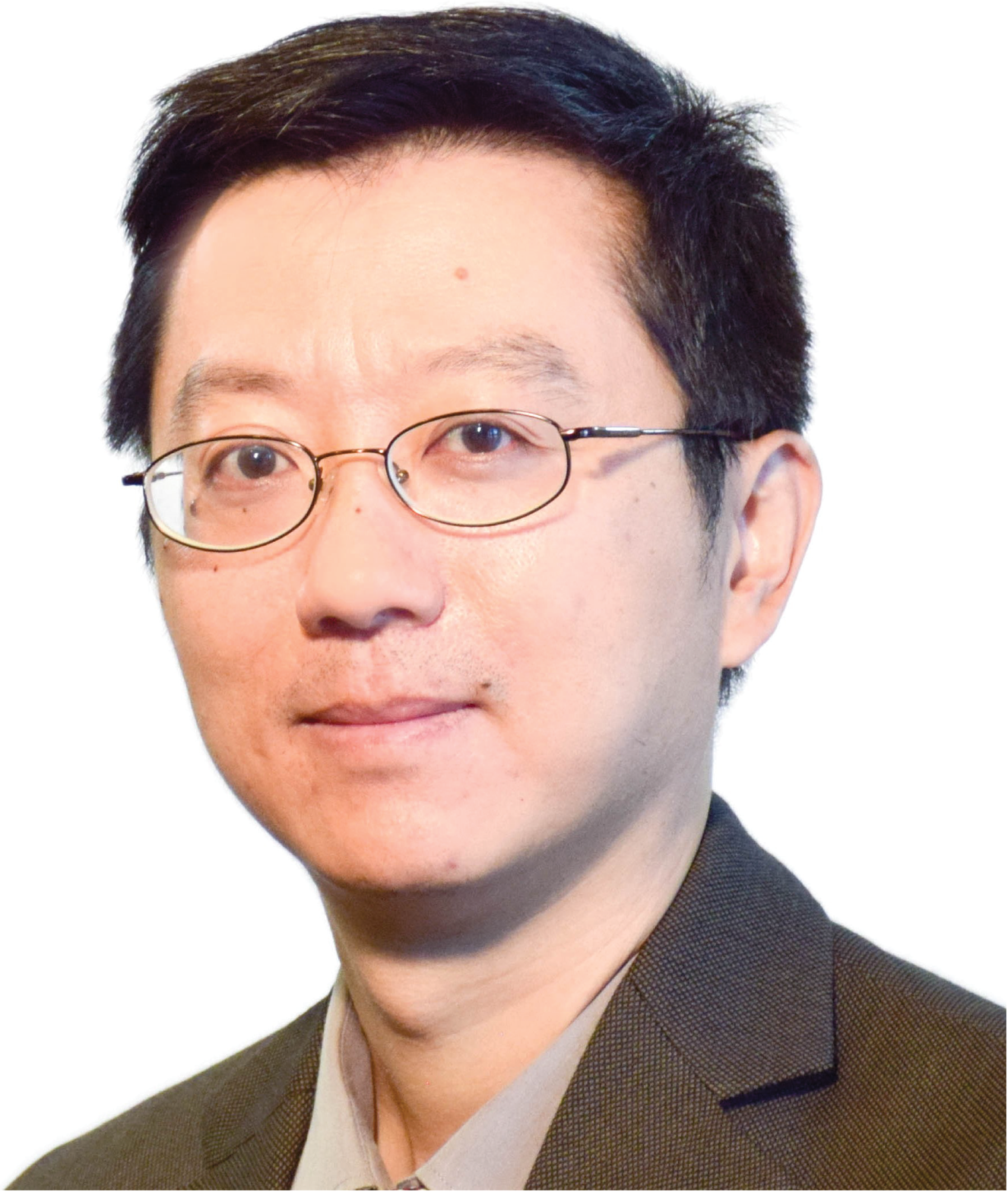}}]{Huaiyu~Dai} (F'17) received the B.E. and M.S. degrees in electrical engineering from Tsinghua University, Beijing, China, in 1996 and 1998, respectively, and the Ph.D. degree in electrical engineering from Princeton University, Princeton, NJ in 2002.

He was with Bell Labs, Lucent Technologies, Holmdel, NJ, in summer 2000, and with AT\&T Labs-Research, Middletown, NJ, in summer 2001. He is currently a Professor of Electrical and Computer Engineering with NC State University, Raleigh, holding the title of University Faculty Scholar. His research interests are in the general areas of communication systems and networks, advanced signal processing for digital communications, communication theory, and information theory. His current research focuses on networked information processing and crosslayer design in wireless networks, cognitive radio networks, network security, and associated information-theoretic and computation-theoretic analysis.

He has served as an editor of \textsc{IEEE Transactions on Communications}, \textsc{IEEE Transactions on Signal Processing}, and \textsc{IEEE Transactions on Wireless Communications}. Currently he is an Area Editor in charge of wireless communications for \textsc{IEEE Transactions on Communications}, and a member of the Executive Editorial Committee for \textsc{IEEE Transactions on Wireless Communications}. He co-chaired the Signal Processing for Communications Symposium of IEEE Globecom 2013, the Communications Theory Symposium of IEEE ICC 2014, and the Wireless Communications Symposium of IEEE Globecom 2014. He was a co-recipient of best paper awards at 2010 IEEE International Conference on Mobile Ad-hoc and Sensor Systems (MASS 2010), 2016 IEEE INFOCOM BIGSECURITY Workshop, and 2017 IEEE International Conference on Communications (ICC 2017).
\end{IEEEbiography}

\end{document}